\documentclass{article}

\usepackage{arxiv}

\usepackage[utf8]{inputenc} 
\usepackage[T1]{fontenc}    
\usepackage{hyperref}       
\usepackage{url}            
\usepackage{booktabs}       
\usepackage{amsfonts}       
\usepackage{nicefrac}       
\usepackage{microtype}      

\usepackage{graphicx}
\usepackage{amssymb}
\usepackage{dsfont}
\usepackage{mathtools}
\usepackage{tikz-cd}
\usepackage{subcaption}
\usepackage{amsthm}

\makeatletter
\newcommand\mathcircled[1]{%
  \mathpalette\@mathcircled{#1}%
}
\newcommand\@mathcircled[2]{%
  \tikz[baseline=(math.base)] \node[draw,circle,inner sep=1pt] (math) {$\m@th#1#2$};%
}
\makeatother

\usepackage{hyperref}
    \hypersetup{
        colorlinks   = true,
        citecolor    = blue,
        linkcolor=blue  
    }

\usepackage[square,sort,comma,numbers]{natbib}

\newtheorem{theo}{\textbf{Theorem}}[section]

\newtheorem{prop}[theo]{\textbf{Proposition}}

\title{Interactions between different predator-prey states.\\
A method for the derivation of the functional and numerical response.}

\author{
  Cecilia Berardo\thanks{Corresponding author.} \\
  Department of Mathematics and Statistics\\
  FI-00014 University of Helsinki, Finland\\
  \texttt{cecilia.berardo@helsinki.fi} \\
   ORCID: 0000-0002-1729-3765
   \And
 Stefan Geritz \\
 Department of Mathematics and Statistics\\
 FI-00014 University of Helsinki, Finland\\
 \texttt{stefan.geritz@helsinki.fi} \\
 ORCID: 0000-0002-7865-3541    
 \And
 Mats Gyllenberg\\
 Department of Mathematics and Statistics\\
 FI-00014 University of Helsinki, Finland\\
\texttt{mats.gyllenberg@helsinki.fi}  \\
ORCID: 0000-0002-0967-8454
\And
Ga\"{e}l Raoul \\
Centre de Math\'ematiques Appliqu\'ees\\
\'Ecole Polytechnique, 91128 Palaiseau, France\\
\texttt{gael.raoul@gmail.com} \\
ORCID: 0000-0002-9132-5550 
}

\begin{document}
\maketitle

\begin{abstract}
In this paper we introduce a formal method for the derivation of a predator's functional response from a system of fast state transitions of the prey or predator on a time scale during which the total prey and predator densities remain constant. Such derivation permits an explicit interpretation of the structure and parameters of the functional response in terms of individual behaviour. The same method is also used here to derive the corresponding numerical response of the predator as well as of the prey.
\end{abstract}


\section{Introduction}\label{intro}

The functional response is defined as the average number of prey captured per individual predator per unit of time as a function of the population densities of the prey, the predator or both. Well known examples are the Holling I, II and III \cite{holling1959some} and the Beddington-DeAngelis \cite{deangelis1975model,beddington1975mutual} functional responses. While the Holling type II functional response was derived using a time budget argument, the Holling type III as well as the functional response by Beddington and DeAngelis were introduced without an explicitly modelled underlying mechanism. However, as we show here, these responses (and many more) can be derived from a system of fast state transitions of the prey or predator during which the total prey and predator densities remain constant. \\

For example, Metz and Diekmann \cite{metz2014dynamics} derived the Holling type II functional response assuming two predator states, searching and handling, where the transition from the searching to the handling state is the result of an actual prey capture. As a consequence, the equilibrium distribution of predator densities over the two states depends on the prey density: the higher the prey density, the greater the proportion of the individual predators in the handling state and, since it is only the searching predators that capture the prey, the average number of prey captured per predator per unit of time varies with the prey density exactly as described by Holling. \\

Likewise, the Beddington-DeAngelis functional response, whose traditional interpretation is in terms of predator interference, was derived in a different context by Geritz and Gyllenberg \cite{geritz2012mechanistic}. They assumed two prey states, exposed and hiding, in addition to the two predator states of searching and handling. The transition from the exposed to the hiding state is mediated by the encounter with the predator. The equilibrium density distributions of both the prey and the predator over their respective states, therefore, depend on one another's population density. As searching predators capture only exposed prey, the functional response now is not just a function of the prey density, as in the Holling type II functional response, but also of the density of the predator itself.\\

Equally important as the functional response are the numerical responses of the prey and the predator. One distinguishes between demographic and aggregative numerical responses. The latter is a consequence of individuals moving in space and will not be considered here. The demographic numerical response is the rate of change in population density due to birth and death as a function of the population densities of the prey, the predator or both. The same individual-level processes that determine the functional response can also determine the numerical response. Examples of how the numerical response of the prey can depend on the density of the predator (other than through prey capture) have been given by Geritz and Gyllenberg \cite{geritz2013group,geritz2014deangelis}.\\

Most predator-prey models in the literature are special cases of the model by Gause \cite{gause1934struggle,gause1936further}
\begin{equation}\label{BASIC}
\left\{\begin{array}{l}
\frac{dX}{dt}=g(X)X-f(X)Y\\
\frac{dY}{dt}=\gamma f(X)Y-\delta Y
\end{array}\right.
\end{equation}
where $X$ and $Y$ are the densities of the prey and the predator, respectively, $g(X)$ is the \textit{per capita} growth rate of the prey population if the predator is absent, $f(X)$ is the predator functional response, $\gamma> 0$ is defined as the conversion factor of prey into predator offspring and $\delta>0$ is the \emph{per capita} natural mortality rate of the predator. \\

The predator's numerical response (through birth) in equation (\ref{BASIC}) is given by the term $\gamma f(X)$. However, the linear relationship between the predator's numerical response and its functional response is lost if different prey-handling states have different fertility levels. In this paper we give an example with two predator states (starving and well-fed) where the former hunts to survive (but does not reproduce) while the latter hunts to reproduce. The proportion of predator individuals in each state depends on the prey density such that at low prey densities there are relatively more starving predators. This leads to a nonlinear relation between the numerical and functional responses, which on the population level can be described by a non-constant conversion factor $\gamma(X,Y)$.  \\

Likewise, if searching and handling predator individuals have different death rates
and their relative densities vary with the prey density, then the mean death rate will vary accordingly. Moreover, if two prey states (like exposed and hiding) have different birth or death rates, and if the relative densities of the states depend on the predator density (as in a previous example), then the numerical response of the prey will depend on the predator density in a way that is not directly linked
to prey capture. All in all this leads to the more general model
\begin{equation}\label{BASIC2}
\left\{\begin{array}{l}
\frac{dX}{dt}=g(X,Y)X-f(X,Y)Y\\
\frac{dY}{dt}=\gamma(X,Y) f(X,Y)Y-\delta(X,Y) Y
\end{array}\right.
\end{equation}
Once the states and state transitions have been specified, the expressions for $f$, $g$, $\gamma $ and $\delta $ as functions of both $X$ and $Y$ follow automatically. \\

Functional responses provide a connection between two levels of description of a biological population: the microscopic level, where the interactions between individual behavioural states are described, and a macroscopic level, where only the population size is tracked. The interplay between behavioural states and functional responses has been the focus of numerous research works, see e.g. \cite{pettorelli2015individual,abrams2015ratio,jeschke2002predator,alexander2012functional}. More recently, several teams are trying to understand the impact of stochasticity (due, for instance, to stochastic interactions between predators and prey, or to the limited number of individuals involved). For more details on stochastic models, we refer to \cite{johansson2003local}, while the connection between stochastic and deterministic models is discussed in \cite{dawes2013derivation}. Moreover, in \cite{billiard2018rejuvenating} the first deviation from the deterministic dynamics implied by stochastic effects has been analysed.\\

In the context of global warming and rapid changes of species range, the quantitative information provided by functional responses is a valuable tool to understand population dynamics. Recent studies on invasive species have adopted this approach, for example  \cite{dick2013ecological,barrios2015predator,taylor2018predatory,crookes2019comparative}, while other teams are using this notion to discuss the efficiency of biocontrol agents, as given in \cite{cabral2009predation,schenk2002functional}. In both cases, the precise description provided by the functional response proves key to understanding the effect of the antagonistic behaviour on biodiversity and species density. Other works investigate further the quantitative capabilities of functional responses: they develop fitting methods and algorithms to estimate the parameters of the models, see \cite{pritchard2017frair,skalski2001functional,rosenbaum2018fitting}. \\
 
In this article, we consider that both the predator and the prey population are structured by behavioural states. We give a formal method for the mechanistic derivation of a predator's functional response which provides a clear interpretation of the population dynamics in terms of the underlying individual level processes. In addition, we apply the same method to derive the corresponding numerical response of the predator and the prey as well. \\

Section~\ref{sec2:1} is focused on the time scale separation argument and the general method used to derive the functional and numerical responses. In Section~\ref{sec2:2}, we discuss the existence and uniqueness issues for the fast dynamics that are necessary to apply the time scale separation idea. We illustrate these notions with a canonical example in Section~\ref{sec2:3}. \\

In Sections~\ref{sec:3} and~\ref{sec:4}, we consider two cases where either the prey population or the predator population is structured by behavioural states, while the other species is only described by its total density. An interesting outcome is that the conversion rate and the death rate of the predator population are no longer constant if we suppose that the behavioural states of the predator population impact its reproduction rate or mortality rate. In Section~\ref{sec:5}, we consider a case where both populations are structured by behavioural states. In spite of the more complex interaction structure, we show that it is still possible to understand the fast dynamics of the model and to derive explicitly the functional and numerical responses. In each section, we compare the functional responses that we obtain with the well known functions of Holling and Beddington and DeAngelis.\\

In Section~\ref{end}, we discuss the advantages and drawbacks of our approach.

\section{The general method}\label{sec:2}

\subsection{The model}\label{sec2:1}

Consider a predator-prey model with $x=\large(x_i\large)_{i=1}^{m}$ and $y=\large(y_i\large)_{i=1}^{n}$ where $x_i$ and $y_i$ denote the densities of the prey population and the predator population in the various states. By assuming that the transitions between the different states are fast, we can ignore slower processes such as birth and decay. \\

The ordinary differential equations which model the fast time scale scenario are given by
\begin{equation}\label{SYST}
\left\{\begin{array}{l}
\frac {dx_k}{dt}(t)=\sum_{i=1}^m A_{ki}x_i(t)+\sum_{i=1}^m \left(\sum_{j=1}^nB_{ij}^{(k)}y_j(t)\right)x_i(t),\quad k=1,...,m\\
\frac {dy_k}{dt}(t)=\sum_{i=1}^n \left(\sum_{j=1}^mC_{ij}^{(k)}x_j(t)\right)y_i(t)+\sum_{i=1}^n D_{ki}y_i(t), \quad k=1,...,n
\end{array}\right.
\end{equation}
with the consistency conditions on the parameter values
\begin{eqnarray}\label{cons1}
-A_{ii}&=&\sum_{k=1,k\neq i}^{m} A_{ki},\quad \forall i\in\large[1,m\large]\\ \label{cons2}
-B_{ij}^{(i)}&=&\sum_{k=1,i\neq k}^m B_{ij}^{(k)},\quad \forall i\in\large[1,n\large], j\in\large[1,m\large]\\ \label{cons3}
-D_{ii}&=&\sum_{k=1,k\neq i}^{n} D_{ki},\quad \forall i\in\large[1,n\large]\\ \label{cons4}
-C_{ij}^{(i)}&=&\sum_{k=1,i\neq k}^{n} C_{ij}^{(k)},\quad \forall i\in\large[1,n\large], j\in\large[1,m\large]
\end{eqnarray}
\indent
The prey move from state $i$ to state $k$ with rate $A_{ki}$, with $k,i\in \large[1,m\large]$. They leave state $i$ at rate $A_{ii}$ and enter one of the $k\neq i$ states with $k\in \large[1,m\large]$ at rate $A_{ki}$ with $k\neq i,k\in \large[1,m\large]$, such that (\ref{cons1}) follows.
When interacting with predators in state $j$, prey individuals in state $i$ leave their state at rate $B_{ij}^{(i)}$ and move into one of the $k\neq i$ states, $k\in \large[1,m\large]$ at rate $B_{ij}^{(k)}$ with $k\neq i,k\in \large[1,m\large]$. Then, the consistency condition in (\ref{cons2}) is required.\\
The predators move from state $i$ to state $k$, with $i,k \in\large[1,n\large]$, by interacting with the prey in state $j\in\large[1,m\large]$ with rate $C_{ij}^{(k)}$ or spontaneously with rate $D_{ki}$. As for the prey transitions, the consistency conditions on the parameter values which characterise the interactions between the predator states are given in (\ref{cons3}) and (\ref{cons4}).\\

Note however that the model is not fully general since it does not include prey-prey or predator-predator interactions, nor the formation of complexes of (possibly multiple) prey or predator individuals, such as the formation of groups (see e.g. \cite{geritz2013group}). For example, stalking of the prey and unsuccessful attacks can be modelled as fast reactions, but they involve mixed predator-prey states. These cases fall outside the general framework presented here, as we include only pure prey and pure predator states. However, while including prey-prey or predator-predator interactions as well as mixed predator-prey states is straightforward for concrete applications (see, for example, \cite{jeschke2002predator}), it becomes more difficult to give a general qualitative analysis of the fast dynamics. Moreover, we consider only two separate time scales. The approach can be readily extended to multiple time scales.\\

In order to derive the corresponding functional response on the slow time scale, it is necessary that the fast dynamics settles on a unique hyperbolically stable steady state $(\bf \hat{x},\bf \hat{y})$, where $\bf\hat{x}$ and $\bf\hat{y}$ denote the population column vectors.\\
We calculate the functional response by considering the total number of prey in state $i$ caught with capture rate $\beta_{ij}$ by an individual predator in state $j$ over the size $Y$ of the predator population
\begin{equation}\label{frdef}
f(X,Y)=\frac{\sum_{i=1}^m \sum_{j=1}^n \beta_{ij} \hat{x}_i \hat{y}_j}{Y}
\end{equation}
with $X$ denoting the size of the prey population.\\

Likewise, the prey numerical response can be expressed as
\begin{equation}\label{nrpreydef}
g(X,Y)=\frac{\sum_{i=1}^m \lambda_i \hat{x}_i}{X}-\frac{\sum_{i=1}^m \mu_i \hat{x}_i}{X}
\end{equation}
where $\lambda_i$ and $\mu_i$ are respectively the \textit{per capita} birth and natural mortality rates corresponding to the prey state $i$. Next, for the predator, let $\gamma_{ij}$ and $\delta_j $ denote respectively the \textit{per capita} fecundity of a predator in state $j$ that has captured a prey in state $i$ and the \textit{per capita} natural mortality rate of a predator in state $j$. Then the predator's numerical response is given by
\begin{equation}\label{nrpreddef}
\gamma(X,Y)f(X,Y)-\delta(X,Y)=\frac{\sum_{i=1}^m \sum_{j=1}^n \gamma_{ij}\beta_{ij}\hat{x}_i\hat{y}_j}{Y}-\frac{\sum_{j=1}^n \delta_j \hat{y}_j}{Y}
\end{equation}
where $\gamma(X,Y)=\frac{\sum_{i=1}^m \sum_{j=1}^n \gamma_{ij}\beta_{ij}\hat{x}_i\hat{y}_j}{\sum_{i=1}^m \sum_{j=1}^n \beta_{ij} \hat{x}_i \hat{y}_j}$ and $\delta(X,Y)=\frac{\sum_{j=1}^n \delta_j \hat{y}_j}{Y}$ are respectively the density-dependent conversion factor and mortality rate. \\

The equilibrium on the short time scale gives the frequency distribution of individuals over the various states. This is the same as the distribution of the amount of time that a single individual spends in the various states. Therefore, the population level responses on the long time scale result from time-averaging on the short time scale.\\

Note that the function $\gamma$ does not have a direct interpretation in terms of the individual behaviour. It is a population level model component that we introduce here in order to give the equation in the form of the predator-prey model in (\ref{BASIC}). The functional form of the product of $\gamma$ and $f$, on the other hand, does have an individual level interpretation, which is given in equation (\ref{nrpreddef}).
\\

\subsection{Existence and uniqueness of the fast dynamics equilibrium}\label{sec2:2}

The system in ($\ref{SYST}$) can be rewritten in matrix form as
\begin{equation}\label{MATRSYST}
\left\{\begin{array}{l}
\bf \dot{x}=( A+B(y))x \\
\bf \dot{y}=(C(x)+D)y
\end{array}\right.
\end{equation}
The matrices $\bf A+B(y)$, $\bf A$ and $\bf B(y)$ in $M_m(\mathbb R)$, where we use the notation $M_m(\mathbb R)$ to denote the $m\times m$-matrix space over $\mathbb R$, are non-negative off-diagonal matrices and have negative main diagonal entries. The same conditions apply to the matrices $\bf C(x)+D$, $\bf D$ and $\bf C(x)$  in $M_n(\mathbb R)$.  \\

In the linear case, when $\bf{B}=0,\bf{C}=0$, because of the consistency conditions, the matrices $\bf A$ and $\bf D$ correspond to the transition rate matrices of a continuous time Markov chain and the system in (\ref{MATRSYST}) becomes
\begin{equation}\label{MATRSYSTLIN}
\left\{\begin{array}{l}
\bf \dot{x}=Ax \\
\bf \dot{y}=Dy
\end{array}\right.
\end{equation}
Under the assumption that this Markov chain is irreducible and aperiodic, there exists a unique stationary distribution ${\bf\pi}$, corresponding to the fast dynamics steady state we are looking for and which can be found by solving the system in (\ref{MATRSYST}). Furthermore, the convergence to the limit distribution is exponentially fast. As shown in the example in Section~\ref{sec2:3}, a similar argument is used in the triangular case, when the transitions of one of the two species are not affected by the other population densities.\\

Alternatively, when we consider the non-linear case ($B,C \neq 0$), the existence of the equilibrium corresponding to the fast dynamics is guaranteed by the Perron-Frobenius Theorem and Shauder's Fixed Point Theorem. We give the detailed proof in the Appendix~\ref{app1}. However, the uniqueness of this equilibrium is more difficult to establish.\\

When we consider a model with a small number of states (typically 4 states in total), the uniqueness and hyperbolic stability of the steady states can often be checked directly. This is the case of the application that we will discuss in Section~\ref{sec:5}, where the hyperbolic stability of the fast dynamics equilibrium is verified, as we show in the Appendix~\ref{app2}.\\

If the number of states is larger (more than 4 states in total), we have not been able to prove or refute the uniqueness and hyperbolic stability of the steady states under the assumptions presented in Section~\ref{sec2:1}. If we relax these assumptions, however, we can show that it is possible to build examples where the uniqueness does not hold. In Appendix~\ref{app4}, we construct matrices $\bf A, B, C, D$ that satisfy the assumptions except for some diagonal coefficients of $\bf A$ and $\bf D$ that are equal to $0$ (specifically $A_{11}=A_{nn}=D_{11}=D_{nn}=0$). The fast dynamics has then at least two different steady states. This example, where the uniqueness is an issue, can be seen as a model for an actual biological system. Therefore the uniqueness problem appears not only as a mathematical challenge, but also as an important question for the general application of the method, which first of all requires a good understanding of the fast dynamics asymptotic behaviour.\\

The example discussed in Section~\ref{sec2:3} is a \textit{triangular case} of (\ref{MATRSYST}), where $\bf B(y)=0$ for all $\bf y$, while the applications in Sections~\ref{sec:3} and~\ref{sec:4} model the scenario with $\bf A=0$ (i.e., only a single prey state) and $\bf B(y)=0$ for all $\bf y$ and Section~\ref{sec:5} gives an application with the complete model form.\\

\section{Application: when the transitions of the prey are not due to the interactions with the predator states.}\label{sec2:3}

Consider the following ODE system
\begin{equation}\label{firstex}
\left\{\begin{array}{l}
\frac{dx_k}{dt}=\sum_{i=1}^m A_{ki}x_i(t),\quad k=1,...,m\\ 
\frac{dy_k}{dt}=\sum_{i=1}^n \left(\sum_{j=1}^mC_{ij}^{(k)}x_j(t)\right)y_i(t)+\sum_{i=1}^n D_{ki}y_i(t), \quad k=1,...,n 
\end{array}\right.
\end{equation}
The equations in (\ref{firstex}) do not take into account the movements between different prey states due to the interactions with the predator states, while the predator movements can be induced by the prey states.\\

The first equation can be rewritten as $\dot{\bf x}= \bf A\cdot  x$, where $\bf x$ and $\dot{\bf x}$ are the column vectors of respectively the different prey states $\left(x_i\right)_{i=1}^{m}$ and their derivatives $\left(\frac{dx_i}{dt}\right)_{i=1}^{m}$ and $\bf A$ is the matrix in $M_m(\mathbb R)$ with elements $A_{ij}$.
The elements on the main diagonal are all strictly negative, while the elements off the main diagonal are non-negative. Furthermore, each column sums to zero. The negative diagonal entry represents the lifetime rate of the corresponding state, while the off-diagonal entries are proportional to the transition probabilities of the embedded Markov chain of the continuous time Markov jump process.\\

The equilibrium $\hat{\bf x}$ for the fast dynamics satisfies the equation $ \bf A\cdot \hat{x}=0$.
We denote by $\bf E_m$ the matrix of ones in $M_m(\mathbb R)$. Let $X$ and $Y$ be the total population sizes. Then, $\hat{\bf x}$ satisfies the equation $\bf E_m \cdot \hat{\bf x}=\bar{X}$, where $\bar{\bf X}$ denotes the vector of length $m$ with all elements equal to $X$. Summing up the two equations, we obtain that $\bf \hat{x}=(A+E_m)^{-1}\cdot \bar{X}$.\\
We repeat the same procedure with the system of equations given by $\bf (C(\hat{x})+D)\cdot \hat{y} =0$, where $\hat{\bf y}$ denotes the column vector of length $n$ of the predator states at equilibrium and $\bf (C(\hat{x})+D)$ the matrix in $M_n(\mathbb R)$ evaluated at $\hat{\bf x}$. Furthermore, $\hat{\bf y}$ satisfies $\bf E_n\cdot \hat{y}=\bar{Y}$, where $\bar{\bf Y}$ denotes the vector of length $n$ with all elements equal to the total predator density $Y$. Then, the equilibrium is given by $\bf \hat{y}=(C(\hat{x})+D+E_n)^{-1}\cdot \bar{Y}$.\\

Now we apply the above to the following matrix for the prey states
\begin{equation}
{\bf A}=
\left[\begin{array}{cccccc}
-A_{11} & A_{12} & 0 & 0 & \dots & 0 \\
 A_{21} & -A_{22} & A_{23} & 0 & \dots &  0\\
 0         &   A_{32} &  -A_{33} & \dots & \dots &  0\\
  \vdots  &      &     &  \ddots  &    &   \vdots    \\
 0           &   \dots   & \dots  &   \dots     &  -A_{m-1,m-1}    & A_{m-1,m}\\
 0         &   \dots  &  \dots  &     \dots    &   A_{m,m-1}                        &  -A_{m,m} 
\end{array}\right]
\end{equation}
The matrix is analogous to the transpose of the generator matrix corresponding to a birth-death process where the parameters $A_{k+1,k}$ and $A_{k,k+1}$ are respectively the birth rates and the death rates. In this case, the prey transitions are not predator induced and occur only between neighbouring states. The prey leave the class $x_k$ with rate $A_{kk}=A_{k-1,k}+A_{k+1,k}$. Moreover, the transitions from the $(k-1)$-state and $(k+1)$-state to the $k$-state happen at rates $A_{k,k-1}$ and $A_{k,k+1}$, as shown below:
\begin{equation}
 \begin{tikzcd}[every arrow/.append style={shift left}]
 \mathcircled{x_{k-1}} \arrow{r}{A_{k,k-1}} &\mathcircled{\;\;x_{k\;\;}} \arrow{l}{A_{k-1,k}}  \arrow{r}{A_{k+1,k}}  & \mathcircled{x_{k+1}} \arrow{l}{A_{k,k+1}}
 \end{tikzcd}
\end{equation}
The ordinary differential equations for the $k=2,...,m-1$ states are
\begin{equation}
\frac{dx_k}{dt}=A_{k,k-1}x_{k-1}-A_{kk}x_k+A_{k,k+1}x_{k+1}
\end{equation}
while the fast dynamics of the $k=1,m$ states is modelled by
\begin{eqnarray}\nonumber
\frac{dx_1}{dt}&=&-A_{11}x_1 +A_{12}x_2\\ 
\frac{dx_m}{dt}&=&A_{m,m-1}x_{m-1}-A_{m,m}x_m
\end{eqnarray}
\\
We use the solution of the balance equation for the stationary distribution of the birth-death process.
The equilibrium of the system $\bf \dot{x}=A \cdot x $ is of the form
\begin{equation}
\hat{x}_k=\hat{x}_1 \prod_{i=1}^{k-1} \frac{A_{i+1,i}}{A_{i,i+1}} .
\end{equation}
Given the normalisation condition $\sum_{k=1}^{m} x_k=X$ for the total prey population, we obtain
\begin{equation}
\hat{x}_1=\frac{X}{1+\sum_{k=2}^{m}\prod_{i=1}^{k-1}
\frac{A_{i+1,i}}{A_{i,i+1}}}.
\end{equation}
\\
We now set $A_{i+1,i}=A_iY$, that is, we assume that the prey transitions from each state to the consecutive one are directly proportional to the total predator density. If, for example, consecutive prey states represent increasing levels of protection, then the prey are likely to move from the current state to the next one at higher predator densities. This applies to the context of prey defenses triggered by predator \emph{kairomones}, e.g. chemo-signals which warn the prey of danger (see, for example, \cite{papes2010vomeronasal,apfelbach2005effects,takahashi2005smell}).
We denote $\prod_{i=1}^{k-1}\frac{A_iY}{A_{i,i+1}}$ by $\mathcal{A}_k Y^{k-1}$. Then:
\begin{equation}\label{n1} \hat{x}_1 = \frac{X}{1+\sum_{k=2}^{m}\mathcal A_k Y^{k-1}} \end{equation}
\begin{equation}\label{n2} \hat{x}_k = \hat{x}_1 \mathcal A_k Y^{k-1}. 
\end{equation}\\

Additionally, we assume that the predator has two states: searching and handling. In particular, we assume that the handling state includes every action of the predator that occurs after prey capture, such as the actual killing of the prey, carrying the prey to the lair, eating, digesting, resting and giving birth. On the other hand, the searching state is considered as an highly active state. Therefore, births happen only while the mother is in the handling state, although rarely enough to be negligible on the short time scale in order not to violate the assumption of constant total prey size. The ODE for the searching predators with density $S$ and with attack rates $(c_i)_{i=1}^{m}$ corresponding to each prey state is given by
\begin{equation}
\frac{dS}{dt}=-\left(\sum_{i=1}^{m} c_i \hat{x}_i\right) S + d H
\end{equation}
where $H$ is the density of the handling predators, $\frac{1}{d}$ is the average handling time and $S+H=Y$.\\
The fast dynamics equilibrium for the searching predators and the handling predators is
\begin{equation}\label{p1} \hat{S}=\frac{1}{1+\frac{1}{d} \sum_{i=1}^{m} c_i\hat{x}_i }Y \end{equation}
\begin{equation}\label{p2} \hat{H}=\frac{\frac{1}{d}\sum_{i=1}^{m} c_i x_i}{1+\frac{1}{d}\sum_{i=1}^{m}c_i\hat{x}_i}Y. \end{equation}\\

Since we include prey capture as a fast process, we need that the predator population size is much smaller than the prey population size, i.e. $Y\ll X$, so that the effect of the prey capture on the total prey population size $X$ is negligible. The prey capture is proportional to the predator population size. If we did not make the assumption of rare predators, then the prey would die out on the short timescale. The total population size $Y$ in this example and in the applications in Section \ref{sec:3}, \ref{sec:4} and \ref{sec:5} is always the \emph{magnified} or scaled-up predator population size. In Appendix \ref{newapp} we give the technical details and assumptions about the individual behaviour in order to achieve time scale separation between the fast dynamics and the slow dynamics.\\

The predator functional response $f(X,Y)$ is calculated as in (\ref{frdef}) and given by
\begin{equation}\label{fr}
f(X,Y)=X \cdot \frac{c_1+\sum_{k=2}^{m}c_k \mathcal A_k Y^{k-1}}{1+\sum_{k=2}^{m} \mathcal A_k Y^{k-1}+\frac{1}{d}X\left(c_1+\sum_{k=2}^{m}c_k \mathcal A_k Y^{k-1}\right)}.
\end{equation}\\ 

If $\frac{1}{d}X\left(c_1+\sum_{k=2}^{m}c_k \mathcal A_k Y^{k-1}\right)\gg 1+\sum_{k=2}^{m} \mathcal A_k Y^{k-1}$, the food source is superabundant and an increase in the prey density $X$ does not increase the feeding rate, which reaches a constant saturation level $d$, as in the Holling type II functional response. Then the function in (\ref{fr}) is increasing with the total prey population until it saturates at this value. \\

Furthermore if there is only one prey state the functional response in (\ref{fr}) simplifies to the Holling type II functional response 
\begin{equation} f_1\left(X,Y\right)=\frac{c_1X}{1+\frac{1}{d}c_1 X}. \end{equation}

We recall that in (\ref{fr}) we assume $A_{21} = A_1 Y$, that is the rate which determines the transitions of the prey from the defended state $x_1$ to the exposed state $x_2$ is linearly depending on the total predator density. This interpretation agrees also with the individual behaviours modelled by Geritz and Gyllenberg in \cite{geritz2012mechanistic} for their mechanistic derivation of the Beddington-DeAngelis functional response, where the available prey are in state $x_1$, while the $x_2$ class denotes those individuals that found a refuge from the predators. In particular, in the literature the most common form of the function by Beddington and DeAngelis is given by
\begin{equation}\label{typicalBD}
f(X,Y)=\frac{aX}{1+bX+cY}.
\end{equation}
Here we obtain a generalisation of the Beddington-DeAngelis functional response in (\ref{typicalBD}), which differs from the one in \cite{geritz2012mechanistic} because we suppose that the prey in both states $x_1$ and $x_2$ can be captured but at different rates:
\begin{equation} \hat{x}_1=\frac{A_{12}X}{A_{12}+A_1Y}, \qquad \hat{x}_2=\frac{A_1YX}{A_{12}+A_1Y} \end{equation}
\begin{equation}\label{GENBDE}
f_2(X,Y)=\frac{c_1A_{12}X+c_2A_1XY}{1+\frac{A_1}{A_{12}}Y+c_1\frac{1}{d}X+\frac{c_2A_1}{A_{12}}\frac{1}{d}XY}.
\end{equation}

The graph of the function in (\ref{GENBDE}) is illustrated in Fig.~\ref{fig:sec2.3}. When the prey population $X$ increases, the asymptotic behaviour is the same as described above for the class of functions in (\ref{fr}). At high predator density $Y$, the functional response in (\ref{GENBDE}) with $c_1<c_2$ increases and tends to the Holling type II functional response (Fig.~\ref{fig:a23}). This is in contrast to the Beddington-DeAngelis functional response, which is a decreasing function of $Y$. On the other hand, if $c_1>c_2$ the functional response decreases with the total predator size $Y$ (Fig.~\ref{fig:b23}).\\

If $c_2=0$ (or sufficiently small), the functional response in (\ref{GENBDE}) is a Beddington-DeAngelis functional response of the form given in (\ref{typicalBD}) with $a=c_1A_{12}$, $b=c_1\frac{1}{d}$ and $c=\frac{A_1}{A_{12}}$.\\

If $c_1=0$ (or sufficiently small), then the slow predator dynamics will have an Allee effect, i.e., the predator population cannot grow when its density is below a given threshold. In particular, it would not be able to invade the predator-free population. This is because of the $XY$ term in the numerator of the functional response, which gives a squared $Y$ term for the population level birth rate of the predator. Then, at low predator densities, the predator birth term is dominated by the negative linear term describing predator death. Allee effects in predator-prey systems lead to homoclinic or heteroclinic bifurcations. At the individual level, we assumed that $A_{21}=A_1Y$, that is the rate at which the prey enter state $x_2$ is proportional to the total predator density. If $c_1=0$ and $c_2 > 0$, then the predators consume only the prey in state $x_2$. However, if the total predator density $Y$ is small, there are not enough prey in state $x_2$ for the predators to survive.\\

\begin{figure}[h!]
\centering
\begin{subfigure}[b]{0.45\textwidth}
  \includegraphics[width=5cm,height=3.5cm]{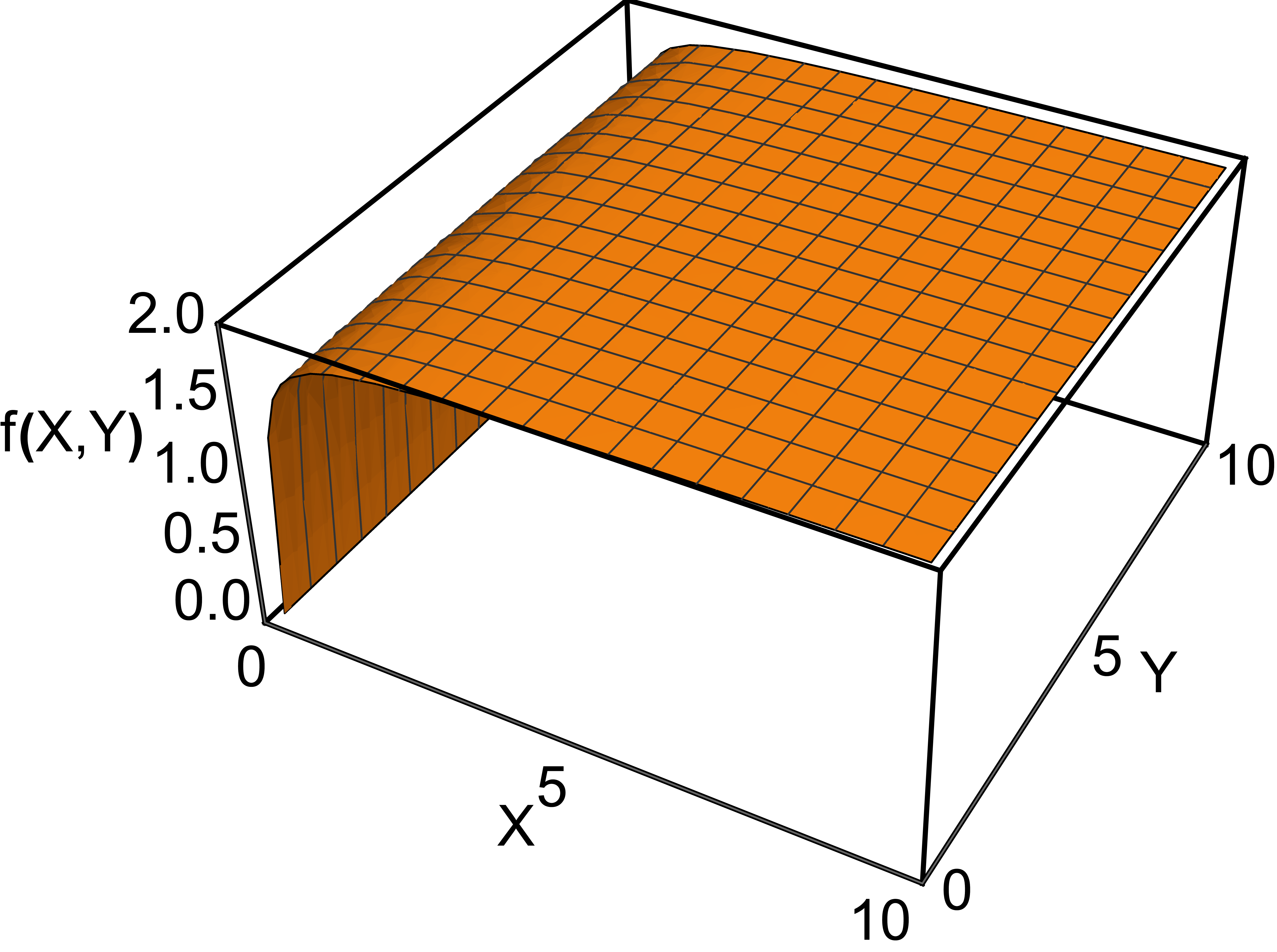}\hspace{0em}
  \caption{}
        \label{fig:a23}
    \end{subfigure}
    \begin{subfigure}[b]{0.45\textwidth}
  \includegraphics[width=5cm,height=3.5cm]{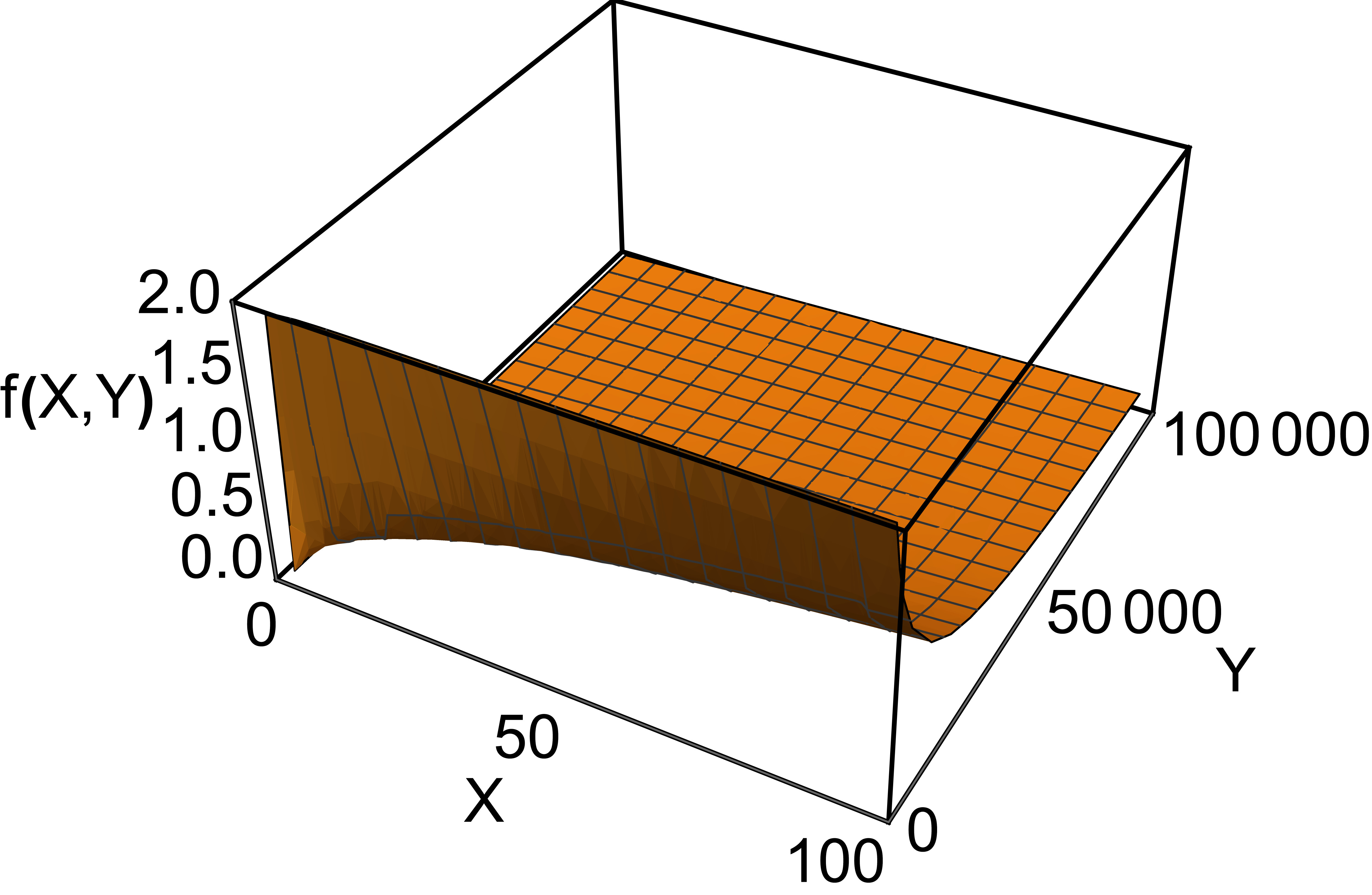}\hspace{0em}
  \caption{}
        \label{fig:b23}
    \end{subfigure}
\caption{The functional response given in (\ref{GENBDE}). Parameter values: (a), $c_1=10$, $c_2=15$, $d=1$, $A_1=0.5$, $A_{21}=2$; (b), $c_1=20$, $c_2=10$, $d=1$, $A_1=0.5$, $A_{21}=2$. }
\label{fig:sec2.3}      
\end{figure}

\section{Application: a functional response with density dependent handling time}\label{sec:3}

\subsection{Individual level reactions, population equations and fast equilibrium}\label{sec3:1}

We analyse the same scenario presented in Section \ref{sec2:3}, that is, the predator population is structured in two classes, the searching predators $S$ and handling predators $H$. In this application we consider only one prey state. We define with $c_1$ the attack rate. We assume moreover that the handlers return to the searching state with prey density dependent rate $c_2X$ or spontaneously with rate $d$. This assumption is ecologically reasonable if the uptake of resources from the corpse of the killed prey declines with the handling time. The capture of a new prey becomes then worthwhile especially if the prey density is high and an indicator of the overall prey density is given by the average time until a new prey gets into the handling predator's field of vision. Furthermore, the density dependent transition may be the result of an actual encounter with a prey individual or may be triggered by prey kairomones (as assumed in the previous section, but with the roles of prey and predator reversed). All these interactions are fast time processes in comparison to birth and natural death and are summarised below:
\begin{eqnarray}\nonumber
\mathcircled{S}& +\;\mathcircled{X} &  \xrightarrow{c_1} \mathcircled{H} \quad  \textit{the searching predator enters the handling state (prey capture)}\\ \nonumber
\mathcircled{H}& +\;\mathcircled{X} &  \xrightarrow{c_2} \mathcircled{S}+ \mathcircled{X} \quad \textit{the predator quits handling with prey-dependent rate}\\
\mathcircled{H}& \xrightarrow{\;\;d\;\;}& \mathcircled{S} \quad \textit{the predator quits handling spontaneously}\\ \nonumber
\end{eqnarray}
Note that in the first transition the prey disappears due to prey capture, while in the second the prey acts merely as a catalyst.\\

The corresponding population level differential equations of the fast time dynamics are given by applying the law of mass action and the time scale separation is presented in details in Appendix \ref{newapp}:
\begin{equation}\label{DEPHTSYST}
\left\{\begin{array}{l}
\frac{dS}{dt}=-c_1XS+c_2XH+dH\\
\frac{dH}{dt}=c_1XS-c_2XH-dH
\end{array}\right.
\end{equation}
The total predator population is constant and given by $Y=S+H$. Then, we can reduce the system of equations to only one equation and solve the steady state equation. The fast dynamics equilibrium is $(\hat{S},\hat{H})=\left(\frac{(d+c_2X)Y}{d+(c_1+c_2)X},\frac{c_1XY}{d+(c_1+c_2)X} \right)$. 

\subsection{Functional response}

We can now derive the corresponding functional response
\begin{equation}\label{DEPHTFR}
f(X)=\frac{c_1X\hat{S}}{Y}=\frac{c_1X(c_2X+d)}{d+(c_1+c_2)X}
\end{equation}
The functional response in (\ref{DEPHTFR}) is a two-parameter function, because only the ratios $\frac{c_1}{d}$ and $\frac{c_2}{d}$ matter and the shape of the functional response is therefore affected by these fractions, as shown in Fig.~\ref{fig:sec3}. Furthermore, by considering the formulation
\begin{equation}\label{DEPHTFRII}
f(X)=\frac{c_1X}{1+\frac{1}{c_2X+d}c_1X},
\end{equation}
we observe that the functional response in (\ref{DEPHTFR}) is like the Holling type II functional response, but with a density dependent handling time given by the ratio $\frac{1}{c_2X+d}$. The higher the prey density is, the faster the predator will quit handling and start searching for fresh food. Note that such behaviour is functional because if the prey is scarce, then the predators will tend to diligently consume the food source until it is completely exhausted and handle the prey longer than if the prey were abundant.

\begin{figure}[h!]
\centering
\begin{subfigure}[b]{0.326\textwidth}
  \includegraphics[width=\textwidth]{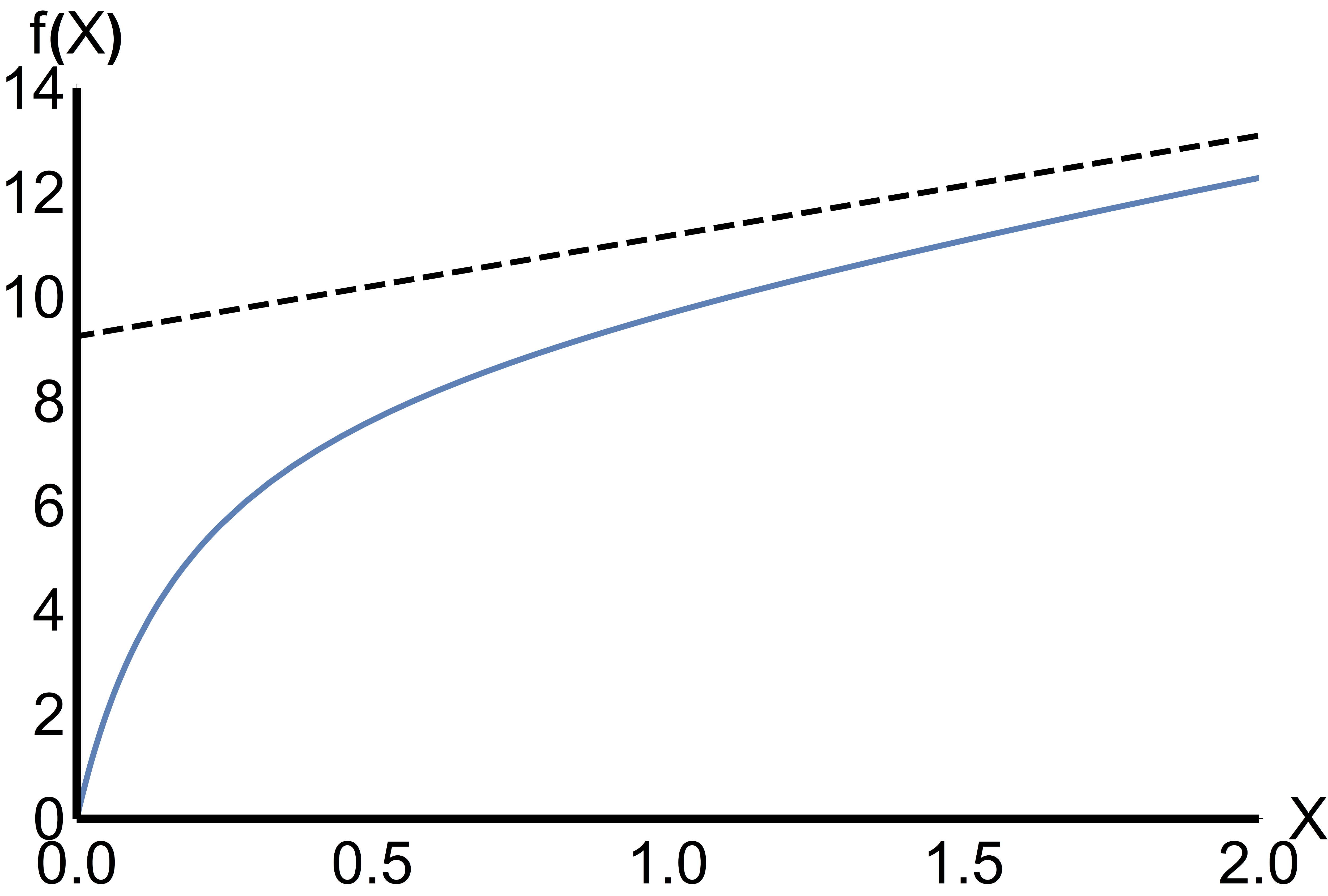}\hspace{0em}
  \caption{}
        \label{fig:a3}
    \end{subfigure}
    \begin{subfigure}[b]{0.326\textwidth}
  \includegraphics[width=\textwidth]{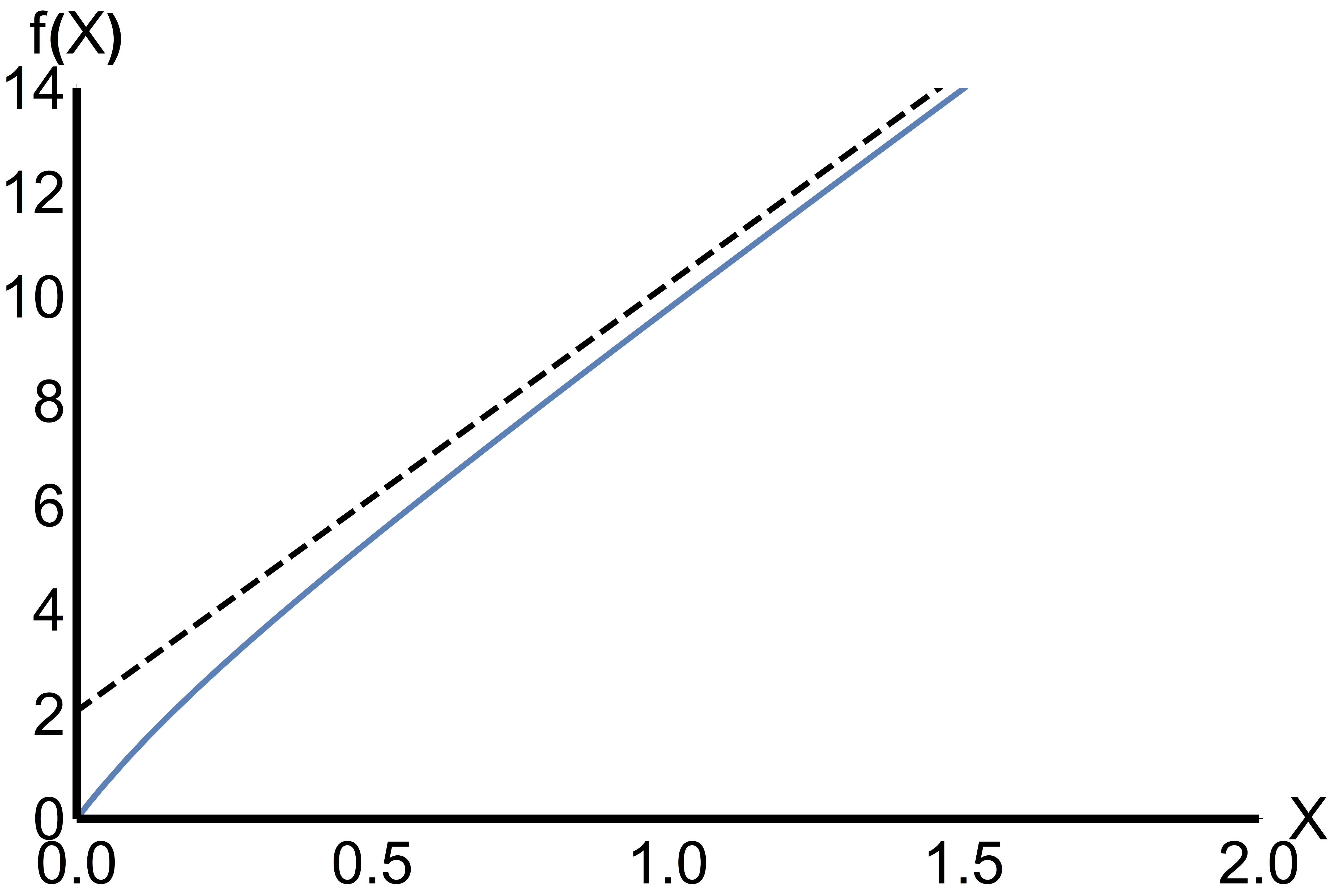}\hspace{0em}
  \caption{}
        \label{fig:b3}
    \end{subfigure}
     \begin{subfigure}[b]{0.326\textwidth}
  \includegraphics[width=\textwidth]{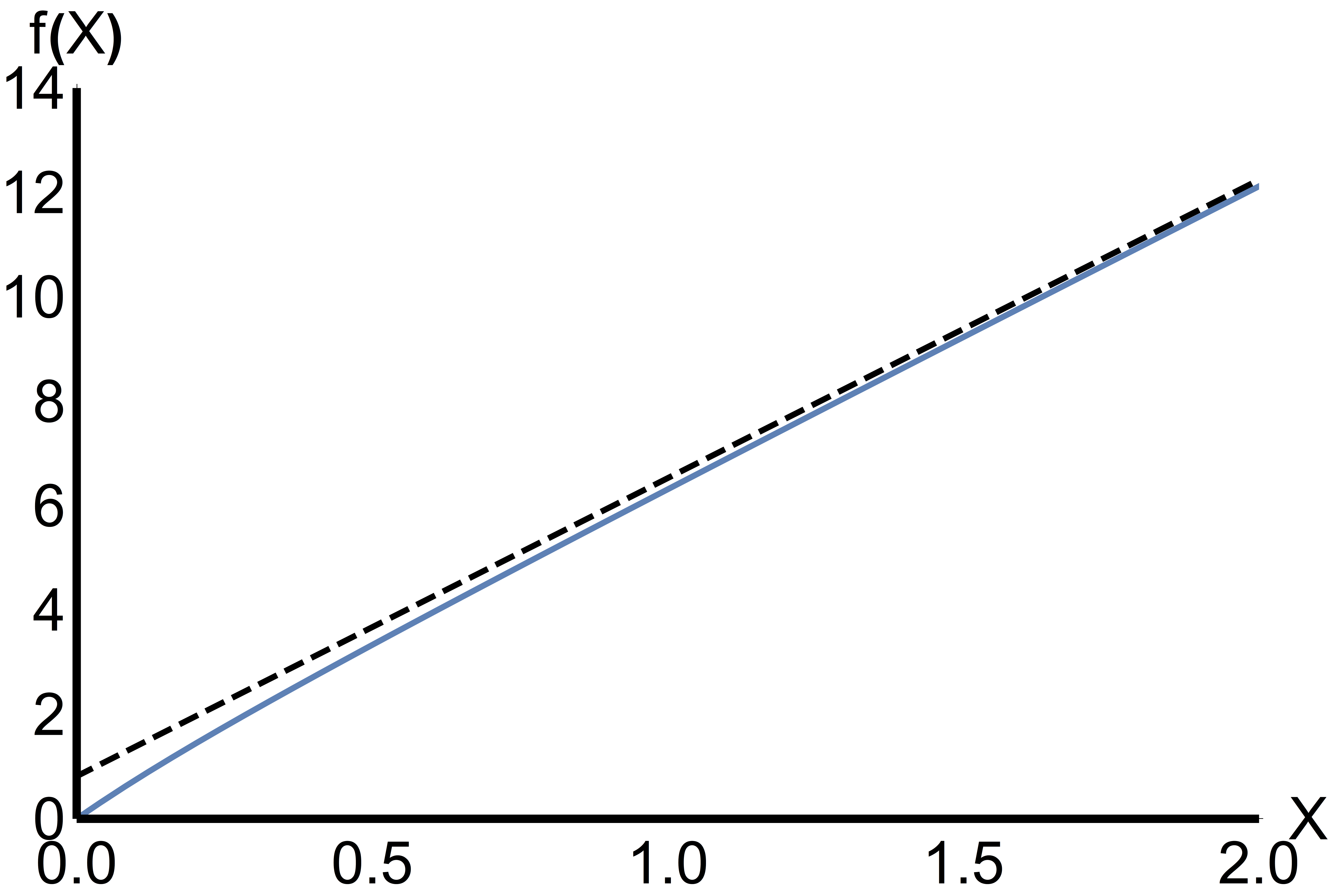}
  \caption{}
        \label{fig:c3}
    \end{subfigure}
\caption{The functional response given in (\ref{DEPHTFR}) here presented for $d=1$ and the following values for $c_1$ and $c_2$: (a), $c_1=50$, $c_2=2$; (b), $c_1=15$, $c_2=18$; (c), $c_1=8$, $c_2=20$. The dashed lines represent the asymptotes corresponding to each curve.}
\label{fig:sec3}      
\end{figure}

\subsection{Predator numerical response}
We assume that the predators in the handling state give birth. Therefore the \emph{per capita} reproduction rate of the predators is proportional to
\begin{equation}
\frac{\hat{H}(X)}{Y}=\frac{f(X)}{c_2X+d}=\gamma(X)f(X)
\end{equation}
The conversion factor $\gamma(X)=\frac{1}{c_2X+d}$ is then a decreasing function of the prey density $X$ and in particular it is proportional to the average time spent handling the prey. 
If we assume that in this particular scenario the predator \emph{per capita} natural mortality rate is the same for the predators both in the searching state and in the handling state, then the mortality rate is constant as in (\ref{BASIC}). 

\section{Application: type III functional response and corresponding predator numerical response}\label{sec:4}

\subsection{Individual level reactions, population equations and fast equilibrium}\label{sec4:1}

We now assume that the searching predators are divided into two subclasses according to the level of starvation, well-fed ($S_1$) and starving ($S_2$). It is natural then to assume different capture rates for the two classes, e.g. starving predators have a lower capture rate than satiated predators, $c_1>c_2$. We suppose again the handling predators in class $H$. We denote with $H_1$ the predators that enter state $H$ from state $S_1$ and with $H_2$ the predators that enter $H$ from state $S_2$. The predators in $H_1$ and $H_2$ handle the prey for, on average, $\frac{1}{d_1}$ units of time, then they enter class $S_1$. If a well-fed predator does not capture a prey in on average $\frac{1}{d_2}$ units of time, it enters the state $S_2$. We assume that starving predators have very low \emph{per capita} fecundity since they are hunting to survive and restock their basic energy reserve. On the contrary the well-fed predators invest part of the energy gained from the food source for reproduction. In this case, the \emph{per capita} fecundities for the two types of consumers, namely $\Gamma_1$ and $\Gamma_2$, differ. In particular $\Gamma_1>\Gamma_2$, since in case of starvation the individuals are likely to cease energy allocation to reproduction, see \cite{kooijman2010dynamic}. We assume the offspring to be in state $S_2$. We consider the transitions between the different predator states to be fast processes with respect to birth and death. In addition, we consider the total predator population size considerably smaller than the total prey size. The time scale separation is achieved through the scaling given in Appendix \ref{newapp}. We summarise below the individual level processes:
\begin{eqnarray}\nonumber
\mathcircled{S_1}&+\;\;\mathcircled{X}& \xrightarrow{c_1} \mathcircled{H_1} \quad \textit{the well-fed predator enters the handling state (prey capture)} \\ \nonumber
\mathcircled{S_2}&+\;\;\mathcircled{X}& \xrightarrow{c_2} \mathcircled{H_2} \quad \textit{the starving predator enters the handling state (prey capture)} \\ \nonumber
\mathcircled{H_1}& \xrightarrow{\;d_1\;} &\mathcircled{S_1} \quad \textit{from the handling to the well-fed state}\\ \nonumber
\mathcircled{H_2}& \xrightarrow{\;d_1\;}& \mathcircled{S_1} \quad \textit{from the handling to the well-fed state} \\ \label{TYPEIIIODE2int}
\mathcircled{S_1}& \xrightarrow{\;d_2\;}& \mathcircled{S_2} \quad \textit{from the well-fed state to the starving state}
\end{eqnarray}

The equations that describe the population level fast time dynamics are given by 
\begin{equation}\label{TYPEIIIODE2}
\left\{\begin{array}{l}
\frac{dS_1}{dt}=-c_1 X S_1+d_1 \left(H_1+H_2\right)-d_2 S_1\\
\frac{dS_2}{dt}=-c_2 X S_2 + d_2 S_1\\
\frac{dH_1}{dt}=c_1 X S_1 -d_1 H_1\\
\frac{dH_2}{dt}=c_2  X S_2-d_1 H_2\\
\end{array}\right.
\end{equation}
If the total predator density is constant, then the conservation law $\frac{dY}{dt}=\frac{dS_1}{dt}+\frac{dS_2}{dt}+\frac{dH_1}{dt}+\frac{dH_2}{dt}=0$ holds. The fast dynamics is settled on the asymptotically stable equilibrium 
\begin{eqnarray}\label{TYPEIIIODEss2}
\hat{S}_1=\frac{c_2 X}{d_2\left(1+c_2\frac{1}{d_1}X\right)+c_2X\left(1+c_1\frac{1}{d_1}X\right)}Y\\\nonumber
\hat{S}_2=\frac{d_2}{d_2\left(1+c_2\frac{1}{d_1}X\right)+c_2X\left(1+c_1\frac{1}{d_1}X\right)}Y\\\nonumber
\hat{H}_1=\frac{\frac{c_1c_2}{d_1}X^2}{d_2\left(1+c_2\frac{1}{d_1}X\right)+c_2X\left(1+c_1\frac{1}{d_1}X\right)}Y\\\nonumber
\hat{H}_2=\frac{\frac{c_2d_2}{d_1}X}{d_2\left(1+c_2\frac{1}{d_1}X\right)+c_2X\left(1+c_1\frac{1}{d_1}X\right)}Y
\end{eqnarray}\\

\subsection{Functional response}

The functional response corresponding to the dynamics above is
\begin{equation}\label{TYPEIII}
f(X)=\frac{c_1 X \hat{S}_1 + c_2 X \hat{S}_2}{Y}=\frac{c_2 X (d_2 + c_1 X)}{d_2\left(1+c_2\frac{1}{d_1}X\right)+c_2X\left(1+c_1\frac{1}{d_1}X\right)}
\end{equation}
The functional response in (\ref{TYPEIII}) is an increasing function of $X$ up to a saturating level given by $d_1$, as shown in Fig. \ref{fig:a4}. It is a type III functional response of the form 
\begin{equation}\label{typeIIIgen}
f(X)=\frac{aX+bX^2}{1+cX+dX^2}
\end{equation}
with $a=c_2$, $b=c_1 c_2 \frac{1}{d_2}$, $c=c_2\left(\frac{1}{d_1}+\frac{1}{d_2}\right)$ and $d=\frac{c_1 c_2}{d_1 d_2}$. A necessary condition for the function to be convex in the neighbourhood of $0$ is that the second derivative has to be positive. This is true if and only if $\frac{c_1}{c_2}>1+d_2\frac{1}{d_1}$, that is the attack rate of the well-fed predators, $c_1$, is sufficiently higher than the attack rate of the starving predators, $c_2$. This result is consistent with the biological interpretation of the individual level dynamics that we have provided.\\

In the literature the most common form of the Holling type III functional response is
\begin{equation}\label{typeIIItrad}
f(X)=\frac{bX^2}{1+dX^2}
\end{equation}
The function in (\ref{typeIIIgen}) can be mathematically reduced to the function in (\ref{typeIIItrad}), if we let $c_2\rightarrow 0$ and $c_1\rightarrow\infty$, such that the product $c_1c_2$ stays constant. This can be interpreted as well-fed predators being extremely efficient hunters, while starving predators are very unsuccessful searchers.

\subsection{Predator numerical response}

On the slow time scale, we suppose that the reproduction rate of the predators is proportional to the density of the handling predators at the fast time equilibrium and that the energy intake from the consumption of the prey is partly allocated to reproduction. At an individual level, we consider the following reactions, which happen at a slow time scale with respect to the interactions modelled in (\ref{TYPEIIIODE2int}):
\begin{eqnarray}
\mathcircled{H_1}& \xrightarrow{\Gamma_1} \mathcircled{H_1}+\mathcircled{S_2}& \textit{the well-fed predator in the handling state}\\\nonumber
& & \textit{produces offspring in state $S_2$}\\ \nonumber
\mathcircled{H_2}& \xrightarrow{\Gamma_2} \mathcircled{H_2}+\mathcircled{S_2}& \textit{the starving predator in the handling state}\\\nonumber
&& \textit{produces offspring in state $S_2$}  \nonumber
\end{eqnarray}

The \emph{per capita} reproduction rate in this particular scenario is not a constant, but it is an increasing function of the total prey population size (see Fig. \ref{fig:c4}) with saturating value given by the fecundity rate of the well-fed predator individuals, $\Gamma_1$, at high prey density:
\begin{equation}\label{frgamma}
\frac{\Gamma_1  \hat{H}_{1}+\Gamma_2  \hat{H}_{2} }{Y}=\frac{c_2\frac{1}{d_1}X \left( \Gamma_1c_1 X+\Gamma_2 d_2  \right)}{d_2\left(1+c_2\frac{1}{d_1}X\right)+c_2X\left(1+c_1\frac{1}{d_1}X\right)}=\gamma(X)f(X).
\end{equation}
The conversion factor
\begin{equation}\label{gamma}
\gamma(X)=\frac{\frac{1}{d_1} \left( \Gamma_1c_1 X+\Gamma_2 d_2\right)}{\left(d_2+c_1 X \right)}
\end{equation}
is a function of the prey density (see Fig. \ref{fig:b4}), saturating on the value $\frac{\Gamma_1}{d_1}$ when the food source is abundant. In fact, $\frac{1}{d_1}$ is the average time spent handling the prey and $\Gamma_1$ denotes the \emph{per capita} fecundity of the well-fed predators. Furthermore, the function $\gamma(X)$ is increasing if and only if $\Gamma_1>\Gamma_2$, which is consistent with the biological assumptions given in Section \ref{sec4:1}. \\

Suppose further that the predators in the two searching states differ not only in their capture rates, but also in their respective mortality rates $\delta_1$ and $\delta_2$, with $\delta_2>\delta_1$. The individual level interactions occurring at the slow time scale which model the natural mortality of the predators are then given by
\begin{eqnarray}
\mathcircled{S_1}& \xrightarrow{\delta_1}  \dagger \qquad& \textit{natural death of the well-fed predator}\\\nonumber
\mathcircled{S_2}& \xrightarrow{\delta_2} \dagger \qquad & \textit{natural death of the starving predator}\nonumber
\end{eqnarray}

Under these assumptions we note that the average per capita mortality rate is given by
\begin{equation}\label{DEATH}
\delta(X)=\frac{\delta_1 \hat{S}_1+\delta_2 \hat{S}_2}{Y}=\frac{\delta_1c_2 X+\delta_2 d_2}{d_2\left(1+c_2\frac{1}{d_1}X\right)+c_2X\left(1+c_1\frac{1}{d_1}X\right)}
\end{equation}
Mortality is then no longer constant, but a decreasing function of the prey density, as shown in Fig. \ref{fig:d4}. In particular, in the absence of the prey, the function $\delta$ takes the value of the per capita mortality rate of the starving predators, $\delta_2$.

\begin{figure}[h!]
\centering
\begin{subfigure}[b]{0.45\textwidth}
  \includegraphics[width=\textwidth]{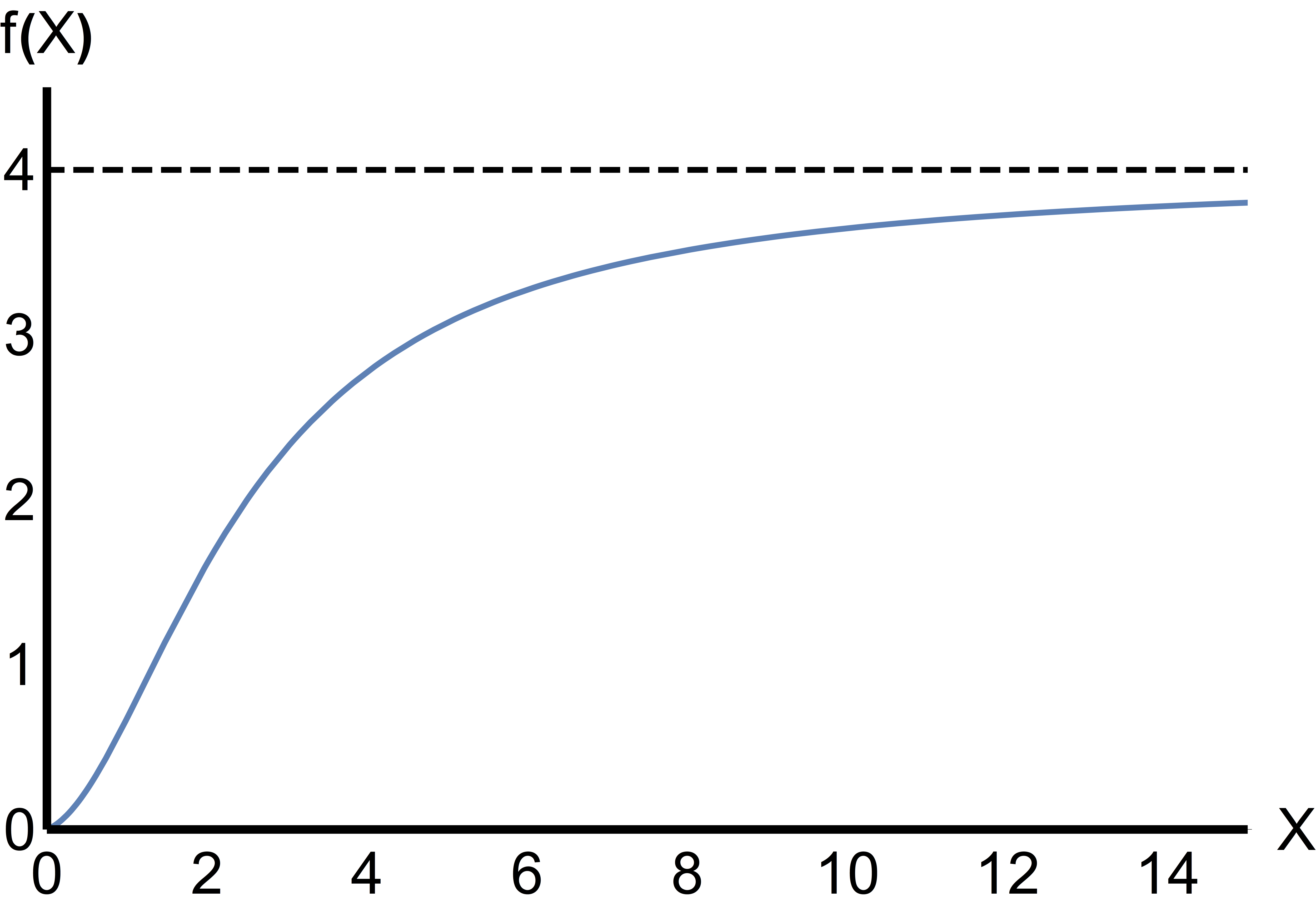}\hspace{0em}
  \caption{}
        \label{fig:a4}
    \end{subfigure}
    \begin{subfigure}[b]{0.45\textwidth}
  \includegraphics[width=\textwidth]{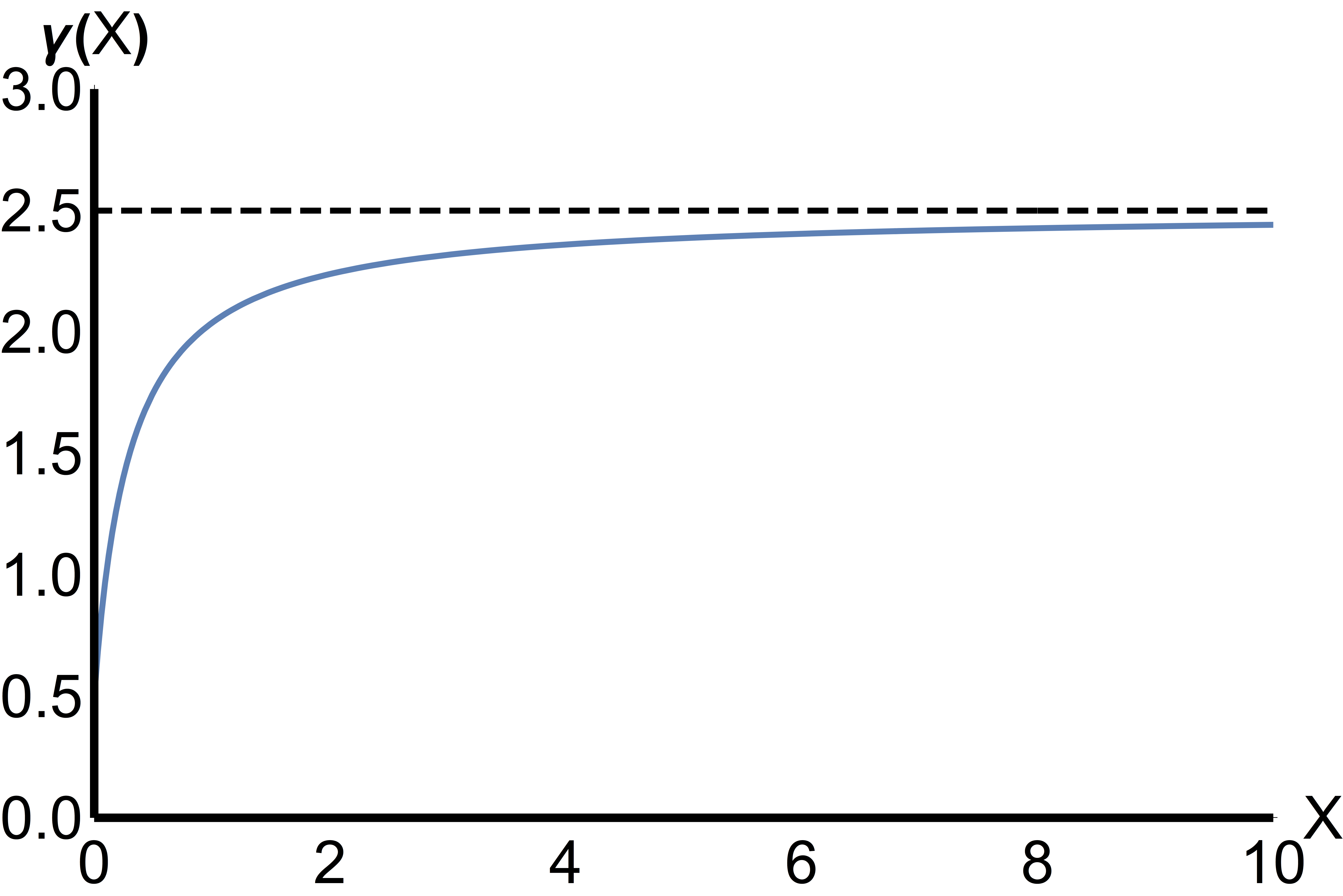}\hspace{0em}
  \caption{}
        \label{fig:b4}
    \end{subfigure}
     \begin{subfigure}[b]{0.45\textwidth}
  \includegraphics[width=\textwidth]{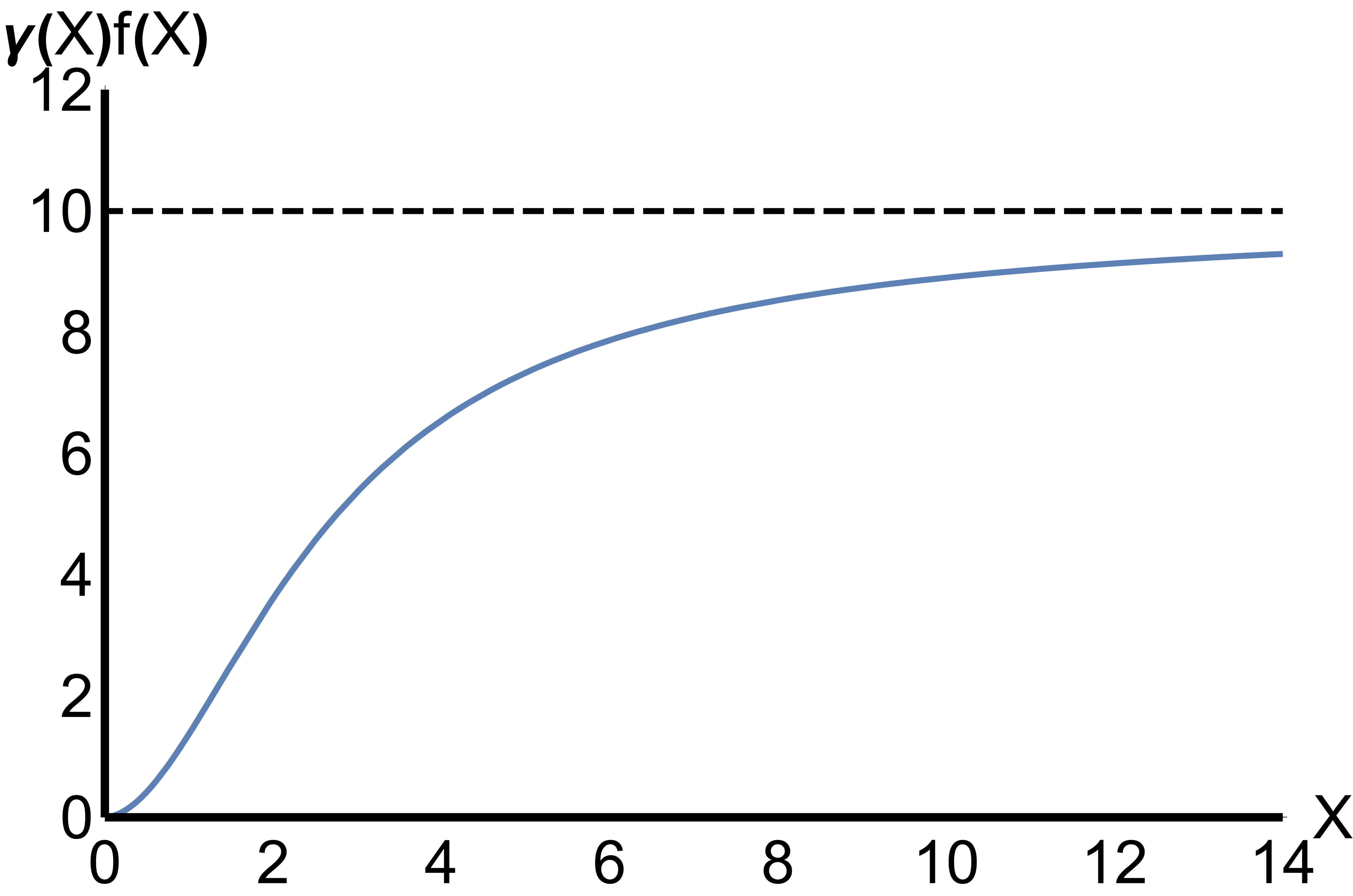}
  \caption{}
        \label{fig:c4}
    \end{subfigure}
    \begin{subfigure}[b]{0.45\textwidth}
  \includegraphics[width=\textwidth]{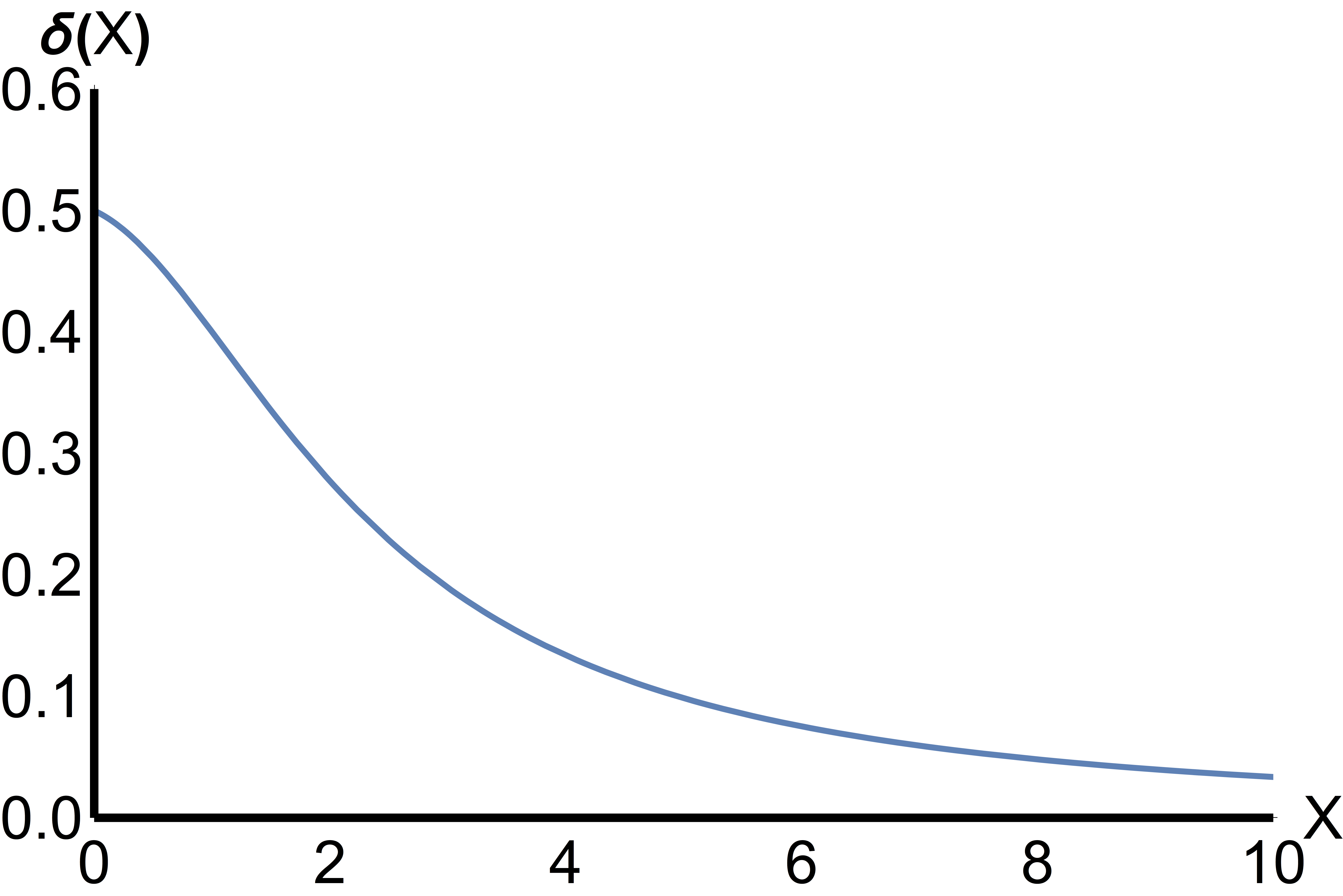}
  \caption{}
        \label{fig:d4}
    \end{subfigure}
\caption{(a): The functional response given in (\ref{TYPEIII}), with saturating value $d_1$ (dashed line). (b): The conversion factor $\gamma(X)$ defined in (\ref{gamma}) with saturating level given by $\frac{\Gamma_1}{d_1}$ (dashed line). Note the positive intercept with the vertical axes, corresponding to the value $\frac{\Gamma_2}{d_1}$. (c): the \emph{per capita} reproduction rate $\gamma(X)f(X)$, derived in (\ref{frgamma}), with asymptote $\Gamma_1$ (dashed line). (d): The \emph{per capita} mortality rate for the predators given in (\ref{DEATH}). Parameter values: $c_1=10$, $c_2=0.2$, $d_1=4$, $d_2=3$, $\Gamma_1=10$, $\Gamma_2=2$, $\delta_1=0.2$, $\delta_2=0.5$.}
\label{fig:sec4}      
\end{figure}

\subsection{Alternative interpretation of the individual level processes}

As far as the functional response is concerned, the model can be interpreted also in the context of predators structured according to their level of experience. This is the usual interpretation of the Holling type III functional response, although until now a rigorous derivation was missing. Here we show how the Holling type III functional response can be derived with our approach. Suppose that class $S_2$ contains the individuals that lack experience and are not well skilled in capturing the prey, while class $S_1$ includes those experienced individuals with success rate $c_1>c_2$. Note that at low prey density, almost all predators will be inexperienced. Predators in both classes, after interaction with the prey, enter the class of handling predators $H$. The average time spent handling the prey is $\frac{1}{d_1}$ units of time. Predators that have captured (and handled) a prey are considered experienced (class $S_1$), but they lose this status, and hence the ability to capture prey at high rate, after on average $\frac{1}{d_2}$ units of time (transition to class $S_2$). Furthermore, we assume that the individual level processes that determine the interactions between the predator and prey states are fast time reactions with respect to birth and death. The above scenario can be visualised in the following way:
\begin{eqnarray}\nonumber
\mathcircled{S_1}&+\;\mathcircled{X}& \xrightarrow{c_1} \mathcircled{H} \quad \textit{the experienced predator enters the handling state (prey capture)}\\ \nonumber
\mathcircled{S_2}&+\;\mathcircled{X}& \xrightarrow{c_2 } \mathcircled{H} \quad \textit{the inexperienced predator enters the handling state (prey capture)}\\ \nonumber
\mathcircled{H}& \xrightarrow{\;d_1\;}& \mathcircled{S_1} \quad \textit{from the handling to the experienced state}\\ 
\mathcircled{S_1}& \xrightarrow{\;d_2\;}& \mathcircled{S_2} \quad \textit{from the experienced to the inexperienced state}
\end{eqnarray}

The ODE system which describes the population level  fast time dynamics is the following
\begin{equation}\label{TYPEIIIODE}
\left\{\begin{array}{l}
\frac{dS_1}{dt}=-c_1X S_1+d_1 H-d_2 S_1\\
\frac{dS_2}{dt}=-c_2 X S_2 + d_2 S_1\\
\frac{dH}{dt}=c_1 X S_1 +c_2  X S_2-d_1 H\\
\end{array}\right.
\end{equation}
The total predator density $Y$ is constant, such that $\frac{dY}{dt}=\frac{dS_1}{dt}+\frac{dS_2}{dt}+\frac{dH}{dt}=0$
Then the fast dynamics settles on the asymptotically stable equilibrium:
\begin{eqnarray}\label{TYPEIIIODEss}
\hat{S}_1=\frac{c_2 X}{d_2\left(1+c_2\frac{1}{d_1}X\right)+c_2X\left(1+c_1\frac{1}{d_1}X\right)}Y\\\nonumber
\hat{S}_2=\frac{d_2}{d_2\left(1+c_2\frac{1}{d_1}X\right)+c_2X\left(1+c_1\frac{1}{d_1}X\right)}Y\\\nonumber
\hat{H}=\frac{c_1c_2\frac{1}{d_1}X^2+c_2d_2\frac{1}{d_1}X}{d_2\left(1+c_2\frac{1}{d_1}X\right)+c_2X\left(1+c_1\frac{1}{d_1}X\right)}Y
\end{eqnarray}
The corresponding functional response has already been given in (\ref{TYPEIII}).

\section{Application: a functional response that induces an Allee effect in the predator population dynamics}\label{sec:5}

\subsection{Individual level reactions, population equations and fast equilibrium}\label{sec5:1}

In the following model, we assume that the prey has a natural tendency to seek protection, but that searching predators are able to overcome the prey defenses, for example by causing them to panic and by attacking the isolated individuals. As shown in \cite{geritz2012mechanistic}, this leads to fast-processes that are the opposite of those in the Beddington-DeAngelis model. The individual level reactions are:
\begin{eqnarray}\nonumber
\mathcircled{E} &\xrightarrow{\;\;a\;\;}& \mathcircled{P} \quad \textit{the exposed prey finds a refuge} \\ \nonumber
\mathcircled{P}&+\;\mathcircled{S} &\xrightarrow{b} \mathcircled{E}+\mathcircled{S} \quad \textit{the protected prey leaves the refuge}\\ \nonumber
\mathcircled{S}&+\;\mathcircled{E}&\xrightarrow{c} \mathcircled{H} \quad \textit{the searching predator enters the handling state (prey capture)}\\ \label{deabereact}
\mathcircled{H} &\xrightarrow{\;\;d\;\;}& \mathcircled{S} \quad \textit{the handling predator quits handling}
\end{eqnarray}

The corresponding differential equations for the fast time population dynamics are
\begin{equation}\label{ode5}
\left\{\begin{array}{l}
\frac {dE}{dt}(t)=bSP-aE\\
\frac{dP}{dt}(t)=-bSP+aE\\
\frac {dS}{dt}(t)=-cES+dH\\
\frac{dH}{dt}(t)=+cES-dH
\end{array}\right.
\end{equation}
with the conservation laws $E+P=X$ and $S+H=Y$ (see Appendix \ref{newapp} for details on the time scale separation between the fast and slow dynamics).
The corresponding fast time equilibrium is given by
\begin{eqnarray}
\hat{E}&=&\frac{1}{2p}\left(-q\left(p+Y\right)+\sqrt{\Delta_{p,q}(X,Y)}\right)\\
\hat{P}&=&\frac{1}{2p}\left(p\left(q+2X\right)+qY-\sqrt{\Delta_{p,q}(X,Y)}\right)\\
\hat{S}&=&\frac{1}{2\left(q+X\right)}\left(-q\left(p-Y\right)+\sqrt{\Delta_{p,q}(X,Y)}\right)\\
\hat{H}&=&\frac{1}{2\left(q+X\right)}\left(pq+qY+2XY-\sqrt{\Delta_{p,q}(X,Y)}\right).
\end{eqnarray}
where $p=\frac{a}{b}$, $q=\frac{d}{c}$ and $\Delta_{p,q}(X,Y)=q\left(p^2q+2p\left(q+2X\right)Y+qY^2\right)$. In the Appendix~\ref{app2}, we give the phase portrait corresponding to the system in (\ref{ode5}), for different values of the parameters, in order to show that the fast dynamics equilibrium is unique and hyperbolically stable.

\subsection{Functional response}\label{sec5:2}

We derive the corresponding functional response (Fig. \ref{fig:a5})
\begin{equation}\label{frrec2}
f(X,Y)= \frac{c\hat{E}\hat{S}}{Y}= \frac{cq}{2\left(q+X \right) Y} \left( pq+qY+2XY- \sqrt{\Delta_{p,q}(X,Y)} \right).
\end{equation}\\

When $p\rightarrow 0$, that is $a\rightarrow 0$ or $b\rightarrow \infty$, the functional response tends to the Holling type II functional response, since the prey is most of the time available for being captured:
\begin{equation}\label{lim1}
f(X,Y)=\frac{cX}{1+\frac{c}{d}X}.
\end{equation}
The case in which $q\rightarrow \infty$, that is $c\rightarrow 0$ or $d\rightarrow \infty$, corresponds to the scenario where the predators are almost all the time searching. Therefore, the predators searching and attacking the prey correspond to the total predators $Y$ and the portion of prey subjected to predation is given by $\frac{bY}{a+bY}X=\frac{Y}{p+Y}X$, with $\frac{bY}{a+bY}$ being the probability for the prey to be in the vulnerable state. Taylor expanding with respect to $\frac{1}{q}$ near zero and retaining only the lowest order term in $\frac{1}{q}$, we get (Fig. \ref{fig:b5})
\begin{equation}\label{lim2}
f(X,Y)=\frac{cXY}{p+Y}.
\end{equation}

We note that if both $q\rightarrow \infty$ and $p\rightarrow 0$, then the functional response in (\ref{frrec2}) becomes linear, as in the Holling type I functional response. In this case the predators handle the prey for an infinitely short time and the prey is most of the time exposed to the predators' attacks. The function is increasing with the attack rate $c$. This is a typical functional response for filter feeders as shown in \cite{jeschke2004consumer}.

\subsection{Prey and predator numerical responses}\label{sec5:3}

If we assume that only the handling predators give birth, then the predator \textit{per capita} birth rate is proportional to 
\begin{equation}\label{birth}
\frac{\hat{H}(X,Y)}{Y}=\frac{f(X,Y)}{d}
\end{equation}
i.e. it is proportional to the functional response, as usual.
On the other hand, when we consider the searching predators $S$ and the handling predators $H$ having different death rates $\delta_1\neq \delta_2$, then the overall \textit{per capita} death rate is 
\begin{equation}\label{death}
\delta_1\frac{\hat{S}(X,Y)}{Y}+\delta_2\frac{\hat{H}(X,Y)}{Y}=\delta_1\left( 1-\frac{f(X,Y)}{d} \right)+\delta_2 \frac{f(X,Y)}{d}
\end{equation}
which is no longer constant, as usually given in the literature, but depends on $X$ and $Y$.\\

We note that the product of $XY$ in the numerator of the functional response in (\ref{lim2}) leads to a squared $Y$ term in the population equations for the prey and the predator. In the predator equation for the slow dynamics, this leads to an Allee effect (i. e. at low predator densities almost all prey are protected and cannot be captured), since the \textit{per capita} birth rate at low predator densities is of order $Y$ (see (\ref{birth})), while, when $q\rightarrow \infty$ and the density of the handling predators is very small, the \textit{per capita} death rate is approximated by the constant value $\delta_1$ (see (\ref{death})). A similar situation was modelled from first principles in \cite{geritz2013group}.\\

If we assume that only the protected prey in class $P$ give birth, i.e., if no other slow interactions between the individual prey or their resources occur, then the prey \textit{per capita} birth rate is proportional to
\begin{equation}
\frac{\hat{P}(X,Y)}{X}=1-\frac{d\cdot f(X,Y)}{cX(d-f(X,Y))}.
\end{equation}
However, when the hiding prey and the available prey have different death rates $\mu_1 \neq \mu_2$ and if there are no other sources of slow death, e.g. interference competition among the prey, then the overall \textit{per capita} death rate is
\begin{equation}
\mu_1 \frac{ \hat{E}(X,Y) }{X}+ \mu_2 \frac{\hat{P}(X,Y)}{X}=\mu_1 \frac{d\cdot f(X,Y)}{cX (d-f(X,Y))}+\mu_2 \left( 1-\frac{d\cdot f(X,Y)}{c X (d-f(X,Y) } \right).
\end{equation}\\

Furthermore, we note that by visualising the prey and predator numerical responses in terms of the functional response, it is possible to understand what the former look like in the limiting cases of the latter as treated above in (\ref{lim1}) and (\ref{lim2}).

\begin{figure}[h!]
\centering
\begin{subfigure}[b]{0.45\textwidth}
  \includegraphics[width=\textwidth]{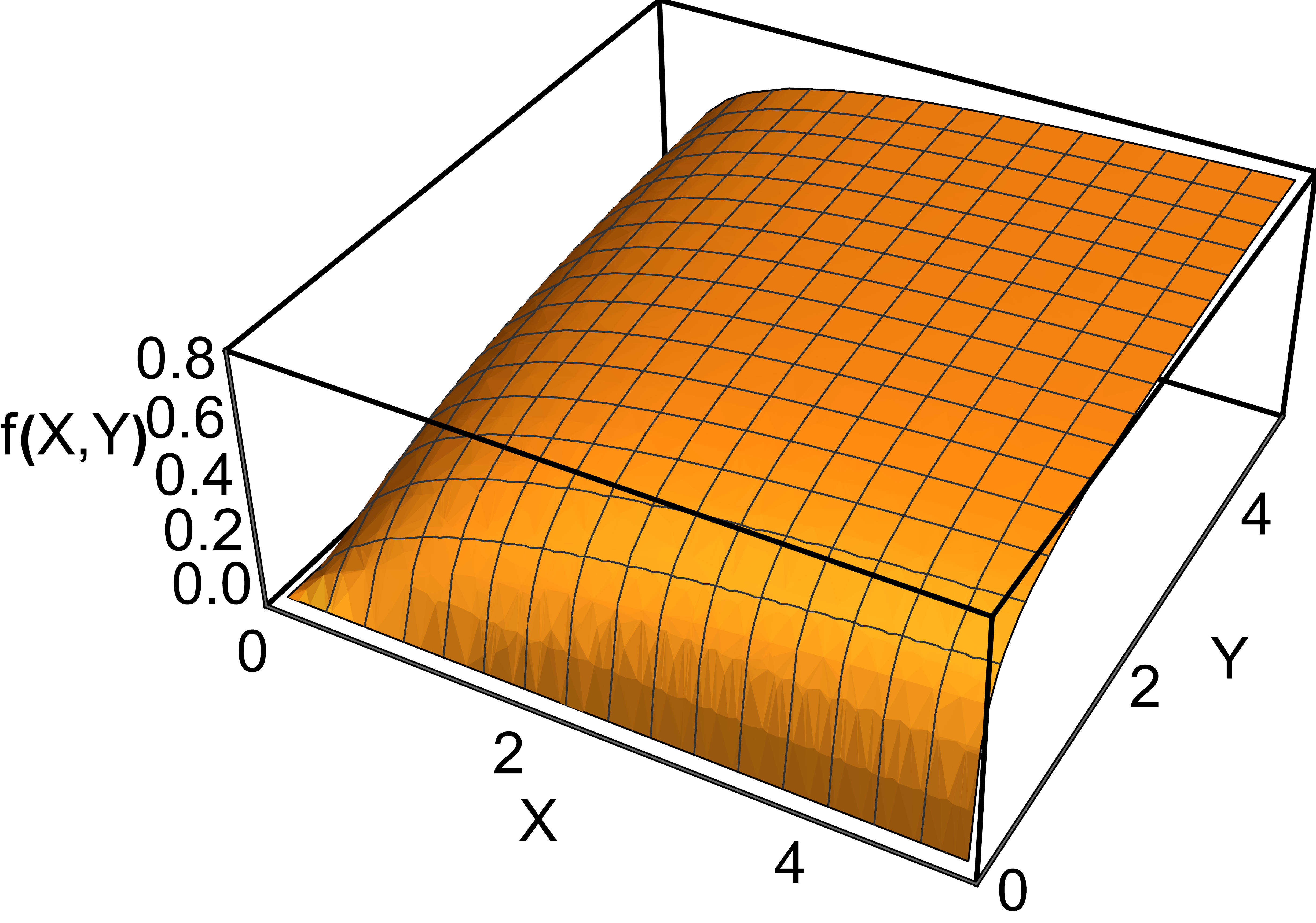}\hspace{0em}
  \caption{}
        \label{fig:a5}
    \end{subfigure}
    \begin{subfigure}[b]{0.45\textwidth}
  \includegraphics[width=\textwidth]{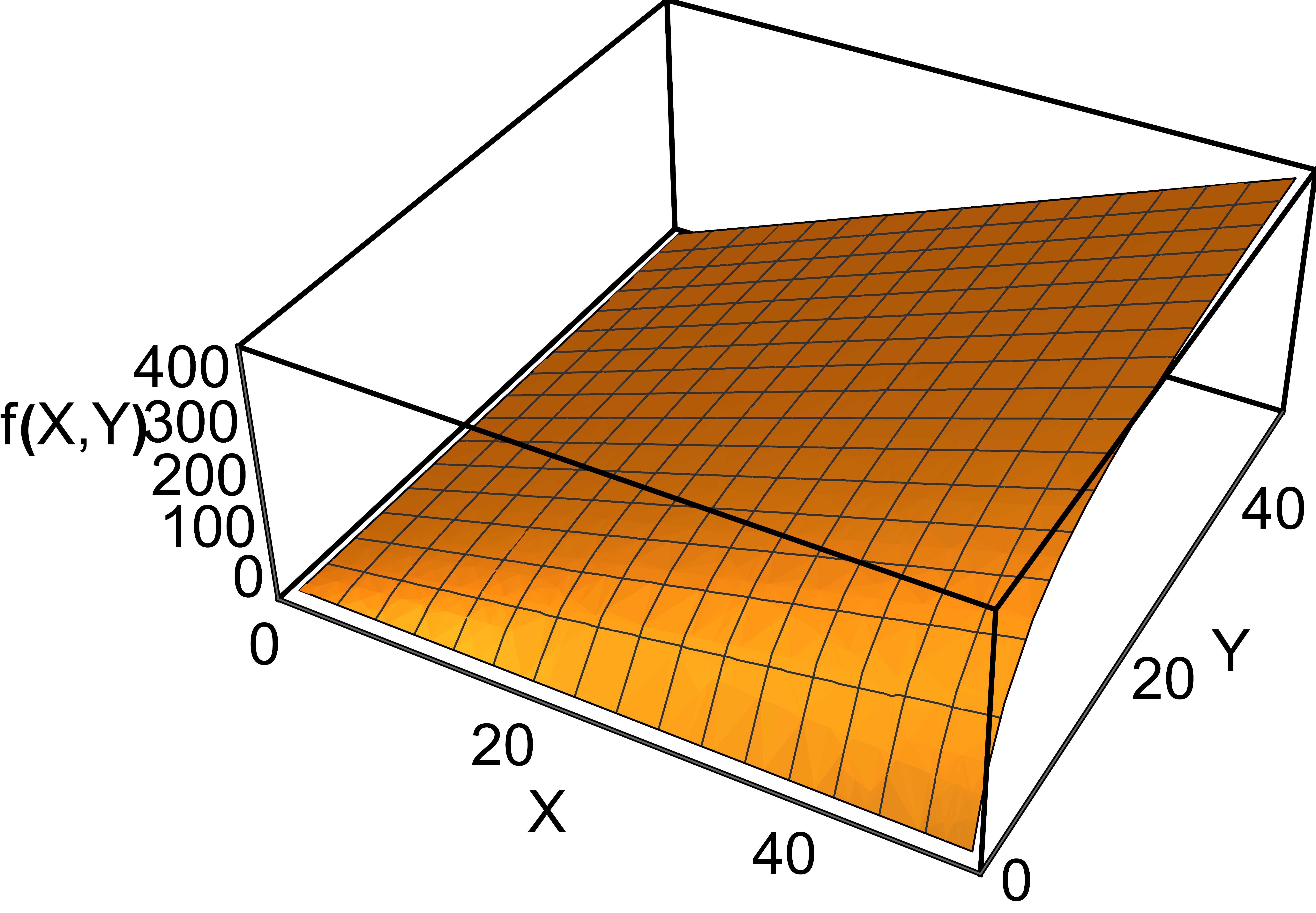}\hspace{0em}
  \caption{}
        \label{fig:b5}
    \end{subfigure}
\caption{(a): The functional response given in (\ref{frrec2}). (b): The functional response defined in (\ref{lim2}). Parameter values: $a=1$, $b=0.2$, $c=10$, $d=1$.}
\label{fig:sec5}      
\end{figure}

\section{Conclusions}\label{end}
In this paper, we proposed a method for the derivation of the functional response from a system of prey-predator interactions which occur on a fast time scale, with respect to birth and death. Many functional responses appear in the literature, but they often lack of interpretation at the individual level. The time scale separation argument that we use in this paper is a possible approach to link the macroscopic behaviour of the population to the microscopic dynamics of the state transitions of individuals. Such derivation permits an explicit interpretation of the structure and parameters of the functional response in terms of the individual behaviour. \\

Elements of the two time-scales are implicit in the traditional approach to deriving the Holling type II functional response, as the consumption rate of the predator instantaneously adjusts to the current prey density. Here we formalise the two-time scales approach in a systematic way. Specific instances of this method can be found in the literature, for example in the works by Metz and Diekmann \cite{metz2014dynamics} and Geritz and Gyllenberg \cite{geritz2012mechanistic,geritz2013group}. However, in this paper we embed these instances into a more general and formal framework that, in addition to the predator's functional response, also gives a derivation for the numerical responses of the predator as well as of the prey. \\

In addition to a general outline of the method, we give several concrete applications, including an application that leads to a generalisation of the Holling type III functional response. The functional response Holling type III has been associated with switching between alternative prey depending on their relative abundance. An explicit derivation was given by Van Leeuwen et al. \cite{leeuwen2007population}. Alternatively, the Holling type III functional response can be associated with different hunger states of the searching predators instead of different experience levels. We give here a mechanistic explanation for the latter. The specific form, as found in the literature, is recovered as a limiting case and is easily understood in terms of the explicitly modelled underlying individual behaviour. \\

In another application, the handling predator may abandon its catch if it detects another live prey. This leads to a Holling type II functional response with density-dependent handling time. In particular, both the handling time and the conversion factor are decreasing functions of the prey density. Such behaviour is adaptive if the uptake of resources from the killed prey declines with the handling time.\\

Furthermore, we discuss the functional response corresponding to a simple non-linear system for the fast dynamics, where we consider two states for the predators and for the prey and we are able to compute explicitly the fast dynamics equilibrium. Here the predators may overcome the prey defenses by causing panic among the prey and by attacking the isolated individuals. We model the prey and predator numerical responses by assuming the two species structured by states with different birth and death rates. The results at the population level are consistent with the individual level reactions and show that at low predator densities an Allee effect is likely to appear.\\

The method presented here is not the most general method possible. For example, we did not consider interactions among the prey themselves or the predators themselves like the exchange of information about the presence of prey or predators leading to a change in the motivational state or the state of alertness. Neither did we include states involving more than one individual, such as two predators fighting over a kill, or several prey seeking protection in numbers, or a predator stalking or fighting a prey. It is not difficult to extend the method to include these cases (e.g. see \cite{geritz2013group} where prey groups of different sizes are modelled as different prey states), but it becomes more difficult to prove the existence and, in particular, the uniqueness of an equilibrium of the fast dynamics of the state transitions.\\

In order to apply slow-fast time scale separation, it is necessary that the fast dynamics is settled on a unique and hyperbolically stable steady state. We are not able to give a general result. In particular, the uniqueness of the equilibrium corresponding to the fast dynamics remains an open question. However, by relaxing the conditions on the parameter values, we have built an example in which the fast time steady state is not unique: this may set a limit to the assumptions that we can make on the coefficients of the matrices which model the interactions in order to get a unique hyperbolically stable equilibrium. \\

In our approach only individuals in some specific discrete states are able to reproduce. The proportion of time that an individual spends in these states is equal to the proportion of individuals in such states at the fast time equilibrium. In this way birth is limited by a time-budget. Another (and possibly more realistic) approach would be to model births as energy limited, like in the dynamic energy budget models (see, for example, \cite{kooijman2010dynamic,geritz2014deangelis}).\\

A possible disadvantage of the approach is that to use it in a practical way, one must make assumptions about the transitions between microstates and the model may become parameter heavy. However, as a theoretical tool, the method has potential. One of the key questions in ecology today, raised by Durrett and Levin \cite{durrett1994importance} among others, is how to scale up from the level of individual behaviours in a population to functional responses and dynamics equations at the population level. Deriving functional and numerical responses from the behaviour of the individual prey and predator is important if one wants to go beyond a mere description of the population dynamics to an understanding in terms of the underlying individual level processes. Also the other way around, that is, if one wants to know the effect of certain changes in the behaviour of the individual prey or predator, the derivation of the population model from first principles in terms of individual behaviour is a necessity.

\appendix
\section{Appendix}\label{app1}
\begin{prop}\label{prop:ext}
Let $A,(B_{i,j}^k)_{i,j\in\{1,\dots,m\}}\in M_m(\mathbb R)$ and $(C_{i,j}^k)_{i,j\in\{1,\dots,n\}}, D\in M_n(\mathbb R)$ be matrices with non-negative off-diagonal coefficients. Suppose that $A$ and $D$ are transition matrices such that the linear system in (\ref{MATRSYSTLIN}) has a unique stable equilibrium and $(B_{i,j}^k)_{i,j\in\{1,\dots,m\}}, (C_{i,j}^k)_{i,j\in\{1,\dots,n\}}$ are irreducible matrices respectively for all $y>0, y \in \mathbb R^n$ and for all $ x>0, x \in \mathbb R^m$. Assume, moreover, that all these matrices are transition matrices and that the conservation laws on the total population density $\sum_{i=1}^{m} x_i=X$ and $\sum_{i=1}^{n} y_i=Y$ hold, with $X$ and $Y$ constant.  \\
Then, the system in (\ref{MATRSYST}) has at least one equilibrium point.
\end{prop}

\begin{proof}
Consider the steady state equations 
\begin{equation}\label{SSEQ}
\left\{\begin{array}{l}
\bf (A+B(y))x=0 \\
\bf (C(x)+D)y=0
\end{array}\right.
\end{equation}
In the first set of equations, $\bf A+B(y)$ $\in M_m(\mathbb R)$ is an irreducible nonnegative off-diagonal matrix for all $\bf y>$ $0$, $\bf y$ $\in \mathbb R^n$ and by the Perron-Frobenius Theorem it has a simple dominant nonnegative eigenvalue, that is $0$. Let $\psi(\bf y)$ be the corresponding eigenvector satisfying 
$\sum_{i=1}^{m} \psi(y)_i=X$.\\
In the same way, for the second set of equations we have that $\bf C(x)+D$ $\in M_n(\mathbb R)$ is an irreducible nonnegative off-diagonal matrix for all $\bf x$ $>0$, $\bf x$ $\in \mathbb R^m$ and by the Perron-Frobenius Theorem it has a simple dominant nonnegative eigenvalue, that is $0$.
Let $\bf \phi(x)$ be the corresponding eigenvector satisfying $\sum_{i=1}^{n} \phi(x)_i=Y$.\\
The continuous map $\Phi$ from the compact convex set $\{\vec{x}\in \mathbb R^m : \sum_{i=1}^m x_i =X \}$ in itself, $\Phi=\psi\circ\phi$ has at least one fixed point by the Shauder's Fixed Point Theorem. This shows that the fast dynamics has at least one steady state.
\end{proof}

\section{Appendix}\label{newapp}
We give here the time scale separation for the system in Section \ref{sec2:3} in details. The dynamical system of the interactions modelled in Section \ref{sec2:3} is given by
\begin{equation}\label{dyn1}
\left\{\begin{array}{l}
\frac{dx_k}{d\tau}=\tilde{A}_{k,k-1}x_{k-1}-\tilde{A}_{kk}x_k+\tilde{A}_{k,k+1}x_{k+1}+ \lambda_k x_k-\mu_k x_k-\tilde{c}_k x_k \tilde{S}, \quad k=1,...,m\\
\frac{d\tilde{S}}{d\tau}=-\left( \sum_{k=1}^m \tilde{c}_k x_k \right)\tilde{S} +\tilde{d} \tilde{H} +\Gamma \tilde{H} -\delta \tilde{S}\\
\frac{d\tilde{H}}{d\tau}=+\left( \sum_{k=1}^m \tilde{c}_k x_k \right)\tilde{S}-\tilde{d} \tilde{H} -\delta \tilde{H}\\
\frac{dX}{dt}=g(X,Y) X - \left( \sum_{k=1}^m \tilde{c}_k x_k \right)\tilde{S}\\
\frac{d\tilde{Y}}{d\tau}=\Gamma \tilde{H} -\delta \tilde{Y}
\end{array}\right.
\end{equation}
where $\lambda_k$ and $\mu_k$ are respectively the \emph{per capita} birth and natural mortality rate for the prey in state $k$ as given in (\ref{nrpreydef}), $\Gamma$ is the conversion rate of prey into predators such that $\Gamma H=\gamma(X,Y)f(X,Y)Y$ and $\delta$ is the \emph{per capita} mortality rate of the predators.\\

Let $\varepsilon >0$ be a small and dimensionless scaling parameter. In order to separate the fast and slow dynamics, we define the following scalings for the parameters of fast time interactions and the predator population that is assumed to be much smaller than the prey population: $\tilde{A}_{k,k-1}=\varepsilon^{-1}{A}_{k,k-1}$, $\tilde{A}_{kk}=\varepsilon^{-1}{A}_{kk}$, $\tilde{A}_{k,k+1}=\varepsilon^{-1}{A}_{k,k+1}$, $\tilde{c}_k=\varepsilon^{-1}{c}_k$, $\tilde{d}=\varepsilon^{-1}d$, $\tilde{Y}=\varepsilon Y$, $\tilde{H}=\varepsilon H$, $\tilde{S}=\varepsilon S$. We now give the slow-fast equations corresponding to the system in (\ref{dyn1}) using the scaled parameters:
\begin{equation}\label{slowfast1}
\left\{\begin{array}{l}
\frac{dx_k}{d\tau}=\frac{A_{k,k-1}}{\varepsilon}x_{k-1}-\frac{A_{kk}}{\varepsilon}x_k+\frac{A_{k,k+1}}{\varepsilon}x_{k+1}+ \lambda_k x_k-\mu_k x_k-c_k x_k S, \quad k=1,...,m\\
\frac{dS}{d\tau}=-\left( \sum_{k=1}^m \frac{c_k}{\varepsilon} x_k \right)S+\frac{d}{\varepsilon} H +\Gamma H -\delta S\\
\frac{dH}{d\tau}=+\left( \sum_{k=1}^m \frac{c_k}{\varepsilon} x_k \right)S-\frac{d}{\varepsilon} H -\delta H\\
\frac{dX}{d\tau}=g(X,Y) X - \left( \sum_{k=1}^m \tilde{c}_k x_k \right)S\\
\frac{dY}{d\tau}=\Gamma H -\delta Y
\end{array}\right.
\end{equation}
We introduce the scaled short time $t=\varepsilon^{-1}\tau$ and let $\varepsilon \rightarrow 0$. We give the equations for the dynamics on the fast time scale:
\begin{equation}\label{fast1}
\left\{\begin{array}{l}
\frac{dx_k}{dt}=A_{k,k-1}x_{k-1}-A_{kk}x_k+A_{k,k+1}x_{k+1}, \quad k=1,...,m\\
\frac{dS}{dt}=-\left( \sum_{k=1}^m c_k x_k \right)S+dH \\
\frac{dH}{dt}=+\left( \sum_{k=1}^m c_k x_k \right)S-d H \\
\frac{dX}{dt}=0\\
\frac{dY}{dt}=0
\end{array}\right.
\end{equation}
The variables $X$ and $Y$ are constants on the fast time scale and from the equations in (\ref{fast1}) for the fast variables $x_k$, $S$ and $H$ we can now derive the fast dynamics equilibria given in (\ref{n1}), (\ref{n2}), (\ref{p1}) and (\ref{p2}).\\

The time scale separations for the models in Sections \ref{sec:3} and \ref{sec:4} follow the passages and the scalings given above for the system in Section \ref{sec2:3}. \\

We consider now the dynamical system for the interactions modelled in Section \ref{sec:5}:
\begin{equation}\label{dyn2}
\left\{\begin{array}{l}
\frac{dE}{dt}=-\tilde{a}E + \tilde{b}PS-\tilde{c}ES+\lambda P - \mu_1E\\
\frac{dP}{dt}=+\tilde{a}E- \tilde{b}PS -\mu_2 P\\
\frac{d\tilde{S}}{d\tau}=-\tilde{c}E\tilde{S} +\tilde{d} \tilde{H} +\Gamma \tilde{H} -\delta_1 \tilde{S}\\
\frac{d\tilde{H}}{d\tau}=+\tilde{c}E\tilde{S}-\tilde{d} \tilde{H} -\delta_2 \tilde{H}\\
\frac{dX}{dt}=g(X,Y) X - \tilde{c}E\tilde{S}\\
\frac{d\tilde{Y}}{d\tau}=\Gamma \tilde{H} -\delta(X,Y) \tilde{Y}
\end{array}\right.
\end{equation}
In this case, as in \cite{geritz2012mechanistic}, we additionally assume that $\tilde{b}$ is large in comparison to the other parameters. In this way, the term $-bPS$ in the short time scale equations for the exposed and the protected prey is not negligible, but part of the fast dynamics. We introduce the scaled parameter $\tilde{b}=\varepsilon^{-2}b$. For an alternative scaling, we could assume that the parameters $\tilde{c}$ and $\tilde{d}$ are small in comparison to the other parameters (see also \cite{geritz2012mechanistic}). In this case, we would not need to assume that the total predator size is much smaller than the total prey size, but we would separate the dynamics in (\ref{dyn2}) into three separate time scales.

\section{Appendix}\label{app2}
We compute the solutions for the system of equations in Section~\ref{sec5:1}, for different values of the parameters $a$, $b$, $c$, $d$ and different initial conditions. The numerical simulations in Fig.~\ref{fig:1} show that the uniqueness and hyperbolic stability of the fast dynamics steady state is verified for the chosen values of the parameters.
\begin{figure}[h!]
\centering
\begin{subfigure}[b]{0.326\textwidth}
  \includegraphics[width=\textwidth]{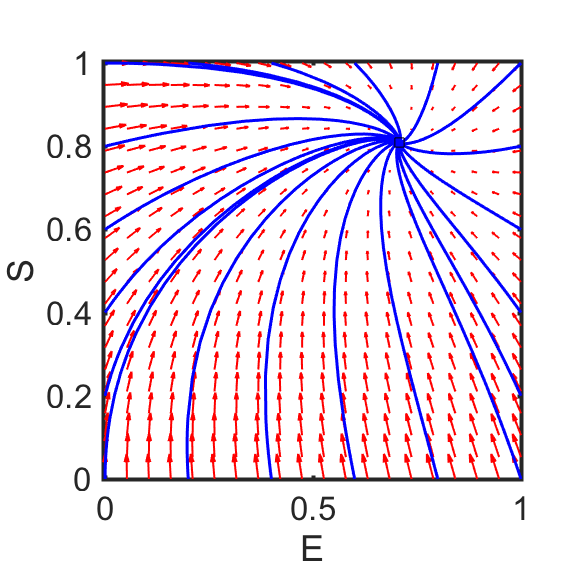}\hspace{0em}
  \caption{}
        \label{fig:a}
    \end{subfigure}
    \begin{subfigure}[b]{0.326\textwidth}
  \includegraphics[width=\textwidth]{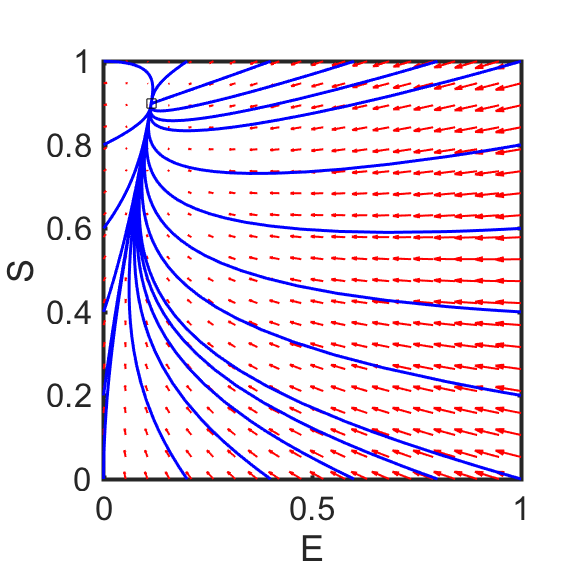}\hspace{0em}
  \caption{}
        \label{fig:b}
    \end{subfigure}
     \begin{subfigure}[b]{0.326\textwidth}
  \includegraphics[width=\textwidth]{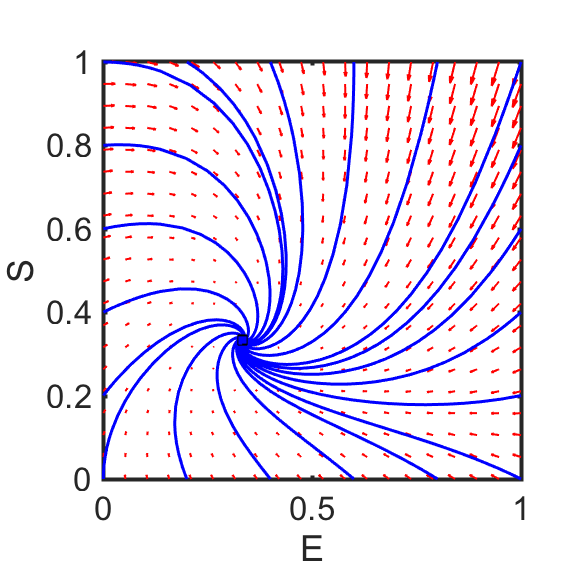}
  \caption{}
        \label{fig:c}
    \end{subfigure}
\caption{Different phase portraits for the system in Section~\ref{sec5:1}. The parameter values are chosen arbitrarily. (a): $a=0.1$, $b=0.3$, $c=0.1$, $d=0.3$; (b): $a=0.7$, $b=0.1$,$c=0.2$, $d=0.2$; (c): $a=0.2$, $b=0.3$, $c=0.6$, $d=0.1$. The total prey and predator densities are both set at $1$.}
\label{fig:1}      
\end{figure}

\section{Appendix}\label{app4}

\subsection{A \textit{counterexample} to the uniqueness of the fast dynamics steady state}

Consider the system in (\ref{SYST}). Suppose $\vec{A},(\vec{B}_{i,j}^k)_{i,j\in\{1,\dots,n\}},(\vec{C}_{i,j}^k)_{i,j\in\{1,\dots,n\}}, \vec{D} \in M_n(\mathbb R)$. \\
Consider weaker assumptions on $(\vec{B}_{i,j}^k)_{i,j\in\{1,\dots,n\}}$ and $(\vec{C}_{i,j}^k)_{i,j\in\{1,\dots,n\}}$, such that they are not non-negative off-diagonal matrices, but verify for every $i,j\in\{1,\dots,n\}$
\begin{equation}
\sum_{k=1}^nB_{i,j}^k=0,\quad \sum_{k=1}^nC_{i,j}^k=0.
\end{equation}
In this particular case, we are able to give numerically a counterexample to the uniqueness of the steady state of the system.\\
We consider the following symmetric system, where $\vec{A}\equiv \vec{D}$ and $(\vec{B}_{i,j}^k)_{i,j\in\{1,\dots,n\}}\equiv (\vec{C}_{i,j}^k)_{i,j\in\{1,\dots,n\}}$:
\begin{equation}
\left\{\begin{array}{l}
\frac {dx_1}{dt}(t)=\left(\frac 3{16}-1\right)x_1+\frac 3{16}x_2+y_1x_1\\ 
\frac {dx_2}{dt}(t)=-\left(\frac 3{16}-1\right)x_1-\frac 3{16}x_2-y_1x_1\\ 
\frac {dy_1}{dt}(t)=\left(\frac 3{16}-1\right)y_1+\frac 3{16}y_2+x_1y_1\\ 
\frac {dy_2}{dt}(t)=-\left(\frac 3{16}-1\right)y_1-\frac 3{16}y_2-x_1y_1 
\end{array}\right.
\end{equation}
This model satisfies the relations necessary for the conservation of mass: one can check that $\frac {d(x_1+x_2)}{dt}(t)=\frac {d(y_1+y_2)}{dt}(t)=0$ (for any $(x_1(t),x_2(t),y_1(t),y_2(t))$). Moreover, the system has two steady states, namely
\begin{equation}
(x_1,x_2,y_1,y_2)=\left(\frac 14,\frac 34,\frac 14,\frac 34\right)
\end{equation}
and
\begin{equation}
(x_1,x_2,y_1,y_2)=\left(\frac 34,\frac 14,\frac 34,\frac 14\right)
\end{equation}
One may check that, thanks to the Implicit Function Theorem, perturbations of the steady states still exist if we make all the coefficients non zero, but close to the coefficients chosen above.\\

Under the stronger assumptions of non-negative off-diagonal matrices $({\bf B}_{i,j}^k)_{i\in\{1,\dots,n\}}$ and $({\bf C}_{i,j}^k)_{i\in\{1,\dots,n\}}$, we are also able to build a counterexample to the uniqueness of the steady state corresponding to the fast dynamics in (\ref{SYST}). In particular, we still assume a symmetric situation where the matrices describing the dynamics for the first set and the second set of equations are equivalent. Moreover, we impose that the matrices corresponding to the linear part of the system of equations are transition matrices such that state $1$ and state $n$ are absorbing states, in a stochastic sense. \\

We construct a counterexample based on the following cross-diffusion system:
\begin{equation} \left\{\begin{array}{l}
\Delta(a(m)n)=0,\\
\Delta(a(n)m)=0,
\end{array}\right.
\end{equation}
with $n,m\in C^2([0,1))$ and Neumann boundary conditions. To show that this system may have several solutions, one may consider $a$ satisfying
\begin{equation} a(\lambda+\cos(\pi x))=\frac 1{\lambda+\cos(\pi(1-x))}.\end{equation}
Then $n=\lambda+\cos(\pi x)$, $m=\lambda+\cos(\pi(1-x))$ is a solution and the symmetry of the system makes it non unique. We have developed a discrete version of this idea and we give in the following section the corresponding numerical example.
We proceed in three steps:\\

\noindent\textbf{Step 1: Defining $\bf x,y,T^x,T^y$ such that $\bf T^yx=0$ and $\bf T^xy=0$.}\\

\noindent\emph{Definition of $\bf x$ and $\bf y$.} \\
We consider $\vec{x},\vec{y} \in \mathbb R^n$ defined by 
\begin{equation}\label{def:XY}
x_i=\lambda+\cos\left(\pi \frac {i-1}{n-1}\right),\quad y_i=\lambda+\cos\left(\pi \frac {n-i}{n-1}\right)
\end{equation}
for every $i\in\{1,\dots,n\}$.
Note that $\bf x$ is different from $\bf y$. \\

\noindent\emph{Definition of $\bf T^{y}$.}\\
Let $\vec{T}^{\vec{y}}\in M_n(\mathbb R)$ a tridiagonal matrix: for any $i\in\{1,\dots,n\}$,
\begin{equation}\label{def:TY}
T^{y}_{i-1,i}=\frac{t^{y}_i}2,\quad T^{y}_{i,i}=-t^{y}_i,\quad T^{y}_{i+1,i}=\frac{t^{y}_i}2,
\end{equation}
except for the two terms $T^{y}_{1,1}=-\frac{t^{y}_1}{2}$, $T^{y}_{2,1}=\frac{t^{Y}_1}{2}$, $T^{y}_{n,n}=-\frac{t^{y}_n}{2}$, $T^{y}_{n-1,n}=\frac{t^{y}_n}{2}$ and all the other coefficients are $0$. In these expressions, we have denoted
\begin{equation}\label{def:tYi}
t^{y}_i=:\frac{1}{\lambda+\cos\left(\pi \frac {i-1}{n-1}\right)}.
\end{equation}
Thanks to the definitions (\ref{def:TY}) and (\ref{def:XY}), for $i\in\{2,\dots,n-1\}$
\begin{equation} (T^{y}x)_i=\frac{t^{y}_{i-1}}2x_{i-1}-t^{y}_ix_i+\frac{t^{y}_{i+1}}2x_{i+1}=0,\end{equation}
while for $i=1$ (a similar computation can be made for $i=n$),
\begin{equation} (T^{y}x)_1=-\frac{t^{y}_1}{2}x_1+\frac{t^{y}_{2}}{2}x_{2}=0.\end{equation}
We have then shown that 
\begin{equation}\label{eq:TYX}
\bf T^{y}x=0.
\end{equation}

\noindent\emph{Definition of $\bf T^{x}$.}\\
Let $\vec{T}^{\vec{x}}\in M_n(\mathbb R)$ a tridiagonal matrix: for any $i\in\{1,\dots,n\}$,
\begin{equation}\label{defTX}
T^{x}_{i-1,i}=\frac{t^{x}_i}2,\quad T^{x}_{i,i}=-t^{x}_i,\quad T^{x}_{i+1,i}=\frac{t^{x}_i}2,
\end{equation}
except for the two terms $T^{x}_{1,1}=-\frac{t^{x}_1}{2}$, $T^{x}_{2,1}=\frac{t^{x}_1}{2}$, $T^{x}_{n,n}=-\frac{t^{x}_n}{2}$, $T^{x}_{n-1,n}=\frac{t^{x}_n}{2}$ and all the other coefficients are $0$. In these expressions, we have denoted
\begin{equation}\label{def:tXi}
t^{x}_i=:\frac{1}{\lambda+\cos\left(\pi \frac {n-i}{n-1}\right)}.
\end{equation}
Then,
\begin{equation}
\bf T^{x}y=0.
\end{equation}

\noindent\textbf{Step 2: Defining a linear interpolation between $\bf T^y$ and $\bf T^x$.}\\

We want to define $\alpha_i$, $\beta_i$ such that 
\begin{equation}\label{eq:interp}
\alpha_i +\beta_i y_{\mu(i)}=t^{y}_i,\quad \alpha_i +\beta_i x_{\mu(i)}=t^{x}_i,
\end{equation} 
where $\mu(i)=1$ if $i\leq \frac n 2$ and $\mu(i)=n$ otherwise. \\

\noindent\emph{Case where $i\leq \frac n2$.}\\

\begin{equation}
\alpha_i +\beta_i (\lambda-1)=t^{y}_i,\quad \alpha_i +\beta_i (\lambda+1)=t^{x}_i,
\end{equation} 
Taking the difference of those two equations leads to
\begin{eqnarray}
-2\beta_i&=&\frac{1}{\lambda+\cos\left(\pi \frac {i-1}{n-1}\right)}-\frac{1}{\lambda+\cos\left(\pi \frac {n-i}{n-1}\right)}\\ \nonumber
&=&\frac{\cos\left(\pi \frac {n-i}{n-1}\right)-\cos\left(\pi \frac {i-1}{n-1}\right)}{\left(\lambda+\cos\left(\pi \frac {i-1}n\right)\right)\left(\lambda+\cos\left(\pi \frac {n-i}{n-1}\right)\right)}
\end{eqnarray}
which shows that $\beta_i$ is positive:
\begin{equation}\label{eq:betai}
\beta_i=\frac{\cos\left(\pi \frac {i-1}{n-1}\right)-\cos\left(\pi \frac {n-i}{n-1}\right)}{2\left(\lambda+\cos\left(\pi \frac {i-1}n\right)\right)\left(\lambda+\cos\left(\pi \frac {n-i}{n-1}\right)\right)}>0.
\end{equation}
Next, we add up the two equations appearing in (\ref{eq:interp}). Thanks to (\ref{def:XY}), we have $x_i+y_i=2\lambda$ and thanks to the definition (\ref{def:tYi}), (\ref{def:tXi}) of $t^{y}_i$ and $t^{x}_i$,
\begin{eqnarray}
2\alpha_i+2\lambda \beta_i&=&t^{y}_i+t^{x}_i\\ \nonumber
&=&\frac{1}{\lambda+\cos\left(\pi \frac {i-1}{n-1}\right)}+\frac{1}{\lambda+\cos\left(\pi \frac {n-i}{n-1}\right)}\\ \nonumber
&=&\frac{2\lambda+\cos\left(\pi \frac {i-1}{n-1}\right)+\cos\left(\pi \frac {n-i}{n-1}\right)}{\left(\lambda+\cos\left(\pi \frac {i-1}{n-1}\right)\right)\left(\lambda+\cos\left(\pi \frac {n-i}{n-1}\right)\right)}\\ \nonumber
&=&\frac{2\lambda}{\left(\lambda+\cos\left(\pi \frac {i-1}{n-1}\right)\right)\left(\lambda+\cos\left(\pi \frac {n-i}{n-1}\right)\right)}.
\end{eqnarray}
The value of $\beta_i$ is given by (\ref{eq:betai}), then
\[
2\alpha_i=\frac{2\lambda}{\left(\lambda+\cos\left(\pi \frac {i-1}{n-1}\right)\right)\left(\lambda+\cos\left(\pi \frac {n-i}{n-1}\right)\right)}-\lambda\frac{\cos\left(\pi \frac {i-1}{n-1}\right)-\cos\left(\pi \frac {n-i}{n-1}\right)}{\left(\lambda+\cos\left(\pi \frac {i-1}n\right)\right)\left(\lambda+\cos\left(\pi \frac {n-i}{n-1}\right)\right)}.
\]
Therefore 
\begin{equation}\label{eq:alphai}
\alpha_i>0,
\end{equation}
for $i\leq \frac n 2$ and $i\neq 1$. \\

\noindent\emph{Case where $i> \frac n2$.}\\
\begin{equation}
\alpha_i +\beta_i (\lambda+1)=t^{y}_i,\quad \alpha_i +\beta_i (\lambda-1)=t^{x}_i,
\end{equation} 
Taking the difference of those two equations leads to
\begin{eqnarray}
2\beta_i&=&\frac{1}{\lambda+\cos\left(\pi \frac {i-1}{n-1}\right)}-\frac{1}{\lambda+\cos\left(\pi \frac {n-i}{n-1}\right)}\\ \nonumber
&=&\frac{\cos\left(\pi \frac {n-i}{n-1}\right)-\cos\left(\pi \frac {i-1}{n-1}\right)}{\left(\lambda+\cos\left(\pi \frac {i-1}n\right)\right)\left(\lambda+\cos\left(\pi \frac {n-i}{n-1}\right)\right)}
\end{eqnarray}
which shows that $\beta_i$ is positive:
\begin{equation}\label{eq:betai}
\beta_i=\frac{\cos\left(\pi \frac {n-i}{n-1}\right)-\cos\left(\pi \frac {i-1}{n-1}\right)}{2\left(\lambda+\cos\left(\pi \frac {i-1}n\right)\right)\left(\lambda+\cos\left(\pi \frac {n-i}{n-1}\right)\right)}>0.
\end{equation}
Next, we add up the two equations appearing in (\ref{eq:interp}). Thanks to (\ref{def:XY}), we have $x_i+y_i=2\lambda$ and from the definitions in (\ref{def:tYi}) and (\ref{def:tXi}) of $t^{y}_i$ and $t^{x}_i$ we obtain
\begin{eqnarray}
2\alpha_i+2\lambda \beta_i&=&t^{y}_i+t^{x}_i\\ \nonumber
&=&\frac{1}{\lambda+\cos\left(\pi \frac {i-1}{n-1}\right)}+\frac{1}{\lambda+\cos\left(\pi \frac {n-i}{n-1}\right)}\\ \nonumber
&=&\frac{2\lambda+\cos\left(\pi \frac {i-1}{n-1}\right)+\cos\left(\pi \frac {n-i}{n-1}\right)}{\left(\lambda+\cos\left(\pi \frac {i-1}{n-1}\right)\right)\left(\lambda+\cos\left(\pi \frac {n-i}{n-1}\right)\right)}\\ \nonumber
&=&\frac{2\lambda}{\left(\lambda+\cos\left(\pi \frac {i-1}{n-1}\right)\right)\left(\lambda+\cos\left(\pi \frac {n-i}{n-1}\right)\right)}.
\end{eqnarray}
The value of $\beta_i$ is given by (\ref{eq:betai}), then
\[
2\alpha_i=\frac{2\lambda}{\left(\lambda+\cos\left(\pi \frac {i-1}{n-1}\right)\right)\left(\lambda+\cos\left(\pi \frac {n-i}{n-1}\right)\right)}-\lambda\frac{\cos\left(\pi \frac {n-i}{n-1}\right)-\cos\left(\pi \frac {i-1}{n-1}\right)}{\left(\lambda+\cos\left(\pi \frac {i-1}n\right)\right)\left(\lambda+\cos\left(\pi \frac {n-i}{n-1}\right)\right)}. 
\]
Therefore 
\begin{equation}\label{eq:alphai}
\alpha_i>0,
\end{equation}
for $i>\frac n 2$ and $i\neq n$. \\

\noindent\textbf{Step 3: Conclusion.}\\

\noindent\emph{Definition of $\bf A,B,C,D$.}\\
Let $\vec{A},\vec{B},\vec{C},\vec{D}\in M_n(\mathbb R)$ tridiagonal matrices. We define the matrix $\bf A$ as follows, using the coefficients $\alpha_i$ defined by (\ref{eq:alphai}) for $i\in\{1,\dots,n\}$:
\begin{equation} A_{i-1,i}=\frac{\alpha_i}2,\quad A_{i,i}=-\alpha_i,\quad A_{i+1,i}=\frac{\alpha_i}2. \end{equation}
Moreover we denote $A_{n,1}=\frac{\alpha_1}2$, $A_{1,n}=\frac{\alpha_n}2$, while the other coefficients are $0$. Thanks to (\ref{eq:alphai}), $\bf A$ is then an off-diagonal non-negative matrix, and 
\begin{equation} \sum_{k=1}^nA_{k,i}=A_{k-1,k}+A_{k,k}+A_{k+1,k}=0.\end{equation}
We define next the matrix $\bf \tilde B$, using the coefficients $\beta_i$ defined by (\ref{eq:betai}) for $i\in\{1,\dots,n\}$:
\begin{equation} \tilde B_{i-1,i}=\frac{\beta_i}2,\quad \tilde B_{i,i}=-\beta_i,\quad \tilde B_{i+1,i}=\frac{\beta_i}2,\end{equation}
Moreover we obtain $\tilde B_{n,1}=\frac{\beta_1}2$, $\tilde B_{1,n}=\frac{\beta_n}2$, while the other coefficients are $0$. Thanks to (\ref{eq:betai}), $\bf \tilde B$ is an off-diagonal non-negative matrix. We define the family of matrices $(\vec{B}_{i,j}^k)_k$ by 
\begin{eqnarray} 
B_{i,j}^k&=&\tilde B_{k,i}\textrm{ if }j=\mu(i), \\ \nonumber
B_{i,j}^k&=&0, \textrm{ otherwise}.
\end{eqnarray}
For any $i\in\{1,\dots,n\}$, $(\vec{B}_{i,j}^k)_{k,j}$ is then an off-diagonal non-negative matrix, and it satisfies for $i\in\{1,\dots,n\}$
\begin{eqnarray} 
\sum_{k=1}^nB_{i,j}^k&=&\sum_{k=1}^n\tilde B_{k,i}=\tilde B_{i-1,i}+\tilde B_{i,i}+\tilde B_{i+1,i}=0\textrm{ if }\mu(i)=j,\\ \nonumber
\sum_{k=1}^nB_{i,j}^k&=&0,\textrm{ otherwise}.
\end{eqnarray}
Finally, we define
\begin{equation}\label{def:CD}
\bf C:=B,\quad D:=A.
\end{equation}
Note that the matrices $\vec{A},\,(\vec{B}^k)_k,\,(\vec{C}^k)_k,\,\vec{D}$ we have constructed satisfy the assumptions given at the beginning of this section.\\

\noindent\emph{Showing that $\bf (x,y)$ and $\bf (y,x)$ are two steady states.}\\
For all $k\in\{1,\dots,n\}$, we use the definition of $\bf A$ and $\bf B$ to obtain
\begin{eqnarray}
&&\sum_{i=1}^n A_{k,i}x_i+\sum_{i=1}^n \left(\sum_{j=1}^nB_{i,j}^ky_j\right)x_i=\\ \nonumber
&=&\sum_{i=1}^n A_{k,i}x_i+\sum_{i=1}^n \left(B_{k,i}y_{\mu(i)}\right)x_i=\sum_{i=k-1}^{k+1} \left(A_{k,i}+B_{k,i}y_{\mu(i)}\right)x_i\\ \nonumber
&=&\left(A_{k,k-1}+B_{k,k-1}y_{\mu(k-1)}\right)x_{k-1}+\left(A_{k,k}+B_{k,k}y_{\mu(k)}\right)x_{k}+\left(A_{k,k+1}+B_{k,k+1}y_{\mu(k+1)}\right)x_{k+1}\\ \nonumber
&=&\frac 12\left(\alpha_{k-1}+\beta_{k-1}y_{\mu(k-1)}\right)x_{k-1}-\left(\alpha_k-\beta_ky_{\mu(k)}\right)x_{k}+\frac 12\left(\alpha_{k+1}+\beta_{k+1}y_{\mu(k+1)}\right)x_{k+1}\\ \nonumber
 &=&\frac 12 t^y_{k-1}x_{k-1}-t^y_kx_k+\frac 12 t^y_{k+1}x_{k+1}= \left(T^yx\right)_k=0.
\end{eqnarray}
Similarly,
\begin{eqnarray}
&&\sum_{i=1}^n A_{k,i}y_i+\sum_{i=1}^n \left(\sum_{j=1}^n B_{i,j}^kx_j\right)y_i=\\ \nonumber
&=&\sum_{i=1}^n A_{k,i}y_i+\sum_{i=1}^n \left(B_{k,i}x_{\mu(i)}\right)y_i=\sum_{i=k-1}^{k+1} \left(A_{k,i}+B_{k,i}x_{\mu(i)}\right)y_i \\ \nonumber
&=&\left(A_{k,k-1}+B_{k,k-1}x_{\mu(k-1)}\right)y_{k-1}+\left(A_{k,k}+B_{k,k}x_{\mu(k)}\right)y_{k}+\left(A_{k,k+1}+B_{k,k+1}x_{\mu(k+1)}\right)y_{k+1}\\ \nonumber
& =&\frac 12\left(\alpha_{k-1}+\beta_{k-1}x_{\mu(k-1)}\right)y_{k-1}-\left(\alpha_k-\beta_kx_{\mu(k)}\right)y_{k}+\frac 12\left(\alpha_{k+1}+\beta_{k+1}x_{\mu(k+1)}\right)y_{k+1}\\ \nonumber
&=&\frac 12 t^x_{k-1}y_{k-1}-t^x_ky_k+\frac 12 t^x_{k+1}y_{k+1}= \left(T^xy\right)_k=0.
\end{eqnarray}
Thanks to the symmetry of the coefficients (see (\ref{def:CD})), we also obtain
\begin{equation}
\sum_{i=1}^n \left(\sum_{j=1}^nC_{i,j}^kx_j\right)y_i+\sum_{i=1}^n D_{k,i}y_i=\sum_{i=1}^n A_{k,i}y_i+\sum_{i=1}^n \left(\sum_{j=1}^n B_{i,j}^kx_j\right)y_i=0,
\end{equation}
\begin{equation}
\sum_{i=1}^n \left(\sum_{j=1}^nC_{i,j}^ky_j\right)x_i+\sum_{i=1}^n D_{k,i}x_i=\sum_{i=1}^n A_{k,i}x_i+\sum_{i=1}^n \left(\sum_{j=1}^nB_{i,j}^ky_j\right)x_i=0. 
\end{equation}

Finally, we have constructed two steady-states $\bf (x,y)$ and $\bf (y,x)$ for the system of differential equations that we consider:
\begin{equation} 
\left\{\begin{array}{l}
0=\sum_{i=1}^n A_{k,i}x_i+\sum_{i=1}^n \left(\sum_{j=1}^nB_{i,j}^ky_j\right)x_i,\\ \nonumber
0=\sum_{i=1}^n \left(\sum_{j=1}^nC_{i,j}^kx_j\right)y_i+\sum_{i=1}^n D_{k,i}y_i,
\end{array}\right.
\end{equation}
and
\begin{equation}
\left\{\begin{array}{l}
0=\sum_{i=1}^n A_{k,i}y_i+\sum_{i=1}^n \left(\sum_{j=1}^n B_{i,j}^kx_j\right)y_i,\\ \nonumber
0=\sum_{i=1}^n \left(\sum_{j=1}^nC_{i,j}^ky_j\right)x_i+\sum_{i=1}^n D_{k,i}x_i.
\end{array}\right.
\end{equation}

\subsection{Numerical example}
We construct a numerical example where we consider the prey and the predators structured into three states. The individual level reactions correspond to the following network and reaction rates
\begin{equation}\label{numexreact}
 \begin{tikzcd}[every arrow/.append style={shift left}]
\mathcircled{x_1} \arrow{r}{\frac{1}{6}y_1} & \mathcircled{x_2} \arrow{l}{\frac{1}{4+\sqrt{2}}}  \arrow{r}{\frac{1}{4+\sqrt{2}}} & \mathcircled{x_3}  \arrow{l}{\frac{1}{6}y_3} \\
  \mathcircled{y_1} \arrow{r}{\frac{1}{6}x_1} & \mathcircled{y_2} \arrow{l}{\frac{1}{4+\sqrt{2}}}  \arrow{r}{\frac{1}{4+\sqrt{2}}} & \mathcircled{y_3}  \arrow{l}{\frac{1}{6}x_3} \\
 \end{tikzcd}
\end{equation}
Note that $x_1$ and $x_3$, $y_1$ and $y_3$ are absorbing states for the transition matrices $\bf A$ and $\bf D$, respectively. A possible biological interpretation of the interactions in (\ref{numexreact}) is given by assuming the prey population in two different locations, in particular the prey individuals in state $x_1$ are in the first location, while the prey individuals in state $x_3$ are in the second location. The searching predators are also divided into those individuals in state $y_1$, which are searching for a prey in the first location, and those in state $y_3$, which are hunting in the second location. When a prey in state $x_1$ meets a predator in state $y_1$ (and similarly for a prey in state $x_3$ meeting a predator in state $y_3$), it goes hiding with rate $\frac{1}{6}$ and with rate $\frac{1}{4+\sqrt{2}}$ it goes back either to state $x_1$ or to state $x_3$. After an encounter with a prey individual in state $x_1$, with probability per unit of time $\frac{1}{6}$ the predator in state $y_1$ starts handling the prey. The same interactions occur for the predators in the second location when they meet a prey state $x_3$.\\

Define $x_i$ and $y_i$, with $i=1,2,3$ and $n=3$ as follows
\begin{eqnarray}
x_1=\lambda+\cos\left(\pi \frac {i-1}{n-1}\right)\big|_{i=1}=\lambda+1 \\ \nonumber
x_2=\lambda+\cos\left(\pi \frac {i-1}{n+1}\right)\big|_{i=2}=\lambda+\frac{\sqrt{2}}{2}\\ \nonumber
x_3=\lambda+\cos\left(\pi \frac {i-1}{n-1}\right)\big|_{i=3}=\lambda-1
\end{eqnarray}
\begin{eqnarray}
y_1=\lambda+\cos\left(\pi \frac {n-i}{n-1}\right)\big|_{i=1}=\lambda-1\\ \nonumber
y_2=\lambda+\cos\left(\pi \frac {i-1}{n+1}\right)\big|_{i=2}=\lambda+\frac{\sqrt{2}}{2}\\ \nonumber
y_3=\lambda+\cos\left(\pi \frac {n-i}{n-1}\right)\big|_{i=1}=\lambda+1
\end{eqnarray}
Note that $\bf x$ is different from $\bf y$.\\
Define the matrices $\bf T^y$ and $\bf T^x$ such that $\bf T^y x=0$ and $\bf T^x y=0$ as follows
\begin{eqnarray}
t_1^y=\frac{1}{x_1}=\frac{1}{\lambda+1}\\ \nonumber
t_2^y=\frac{1}{x_2}=\frac{1}{\lambda+\frac{\sqrt{2}}{2}}\\ \nonumber
t_3^y=\frac{1}{x_3}= \frac{1}{\lambda-1}
\end{eqnarray}
\begin{eqnarray}
t_1^x=\frac{1}{y_1}=\frac{1}{\lambda-1}\\ \nonumber
t_2^x=\frac{1}{y_2}=\frac{1}{\lambda+\frac{\sqrt{2}}{2}}\\ \nonumber
t_3^x=\frac{1}{y_3}= \frac{1}{\lambda+1}
\end{eqnarray}
Then
\begin{equation}
\vec{T}^{\vec{y}}=\left[\begin{array}{ccc}
-\frac{t_1^y}{2} & \frac{t_2^y}{2} & 0\\ 
\frac{t_1^y}{2} & -t_2^y & \frac{t_3^y}{2}\\ 
0 & \frac{t_2^y}{2} & - \frac{t_3^y}{2}
\end{array}\right]
\end{equation}
and 
\begin{equation} (T^{y}x)_i=\frac{t^{y}_{i-1}}2x_{i-1}-t^{y}_ix_i+\frac{t^{y}_{i+1}}2x_{i+1}=0.\end{equation}
We can define $\bf T^x$ in the same way and show that
\begin{equation}(T^{x}y)_i=\frac{t^{x}_{i-1}}2y_{i-1}-t^{x}_iy_i+\frac{t^{x}_{i+1}}2y_{i+1}=0.\end{equation}
Then we define a linear interpolation between $\bf T^y$ and $\bf T^x$. In particular, we define $\alpha_i$, $\beta_i$ such that 
\begin{equation}\label{eq:interp}
\alpha_i +\beta_i y_{\mu(i)}=t^{y}_i,\quad \alpha_i +\beta_i x_{\mu(i)}=t^{x}_i,
\end{equation} 
where $\mu(i)=1$ if $i\leq \frac n 2$ and $\mu(i)=n$ otherwise. \\
From now on, all the numerical values are obtained by taking $\lambda=2$. This choice of the parameter values does not have any particular biological justification.\\
Then, by solving the above system of equations for $\alpha_i$ and $\beta_i$, we get
\begin{eqnarray}
\alpha_1=0 \quad \beta_1=\frac{1}{3}\\ \nonumber
\alpha_2=\frac{2}{4+\sqrt{2}} \quad \beta_2=0 \\ \nonumber
\alpha_3=0 \quad \beta_3=\frac{1}{3}\\ \nonumber
\end{eqnarray}
Then, we can define the following system of ordinary differential equations 
\begin{equation}\label{numericsyst}
\left\{\begin{array}{l}
\frac {dx_1}{dt}(t)= \frac{\alpha_2}{2}x_2-\frac{\beta_1}{2}y_1 x_1=\frac{1}{4+\sqrt{2}}x_2 -\frac{1}{6}y_1 x_1\\ \\
\frac {dx_2}{dt}(t)=-\alpha_2 x_2 + \frac{\beta_1}{2}y_1 x_1+\frac{\beta_3}{2}y_3 x_3= -\frac{2}{4+\sqrt{2}}x_2+\frac{1}{6}y_1 x_1+\frac{1}{6}y_3 x_3\\ \\
\frac {dx_3}{dt}(t)=\frac{\alpha_2}{2}x_2-\frac{\beta_3}{2}y_3 x_3 = \frac{1}{4+\sqrt{2}}x_2 - \frac{1}{6}y_3 x_3 \\ \\
\frac {dy_1}{dt}(t)=\frac{\alpha_2}{2}y_2-\frac{\beta_1}{2}x_1 y_1=\frac{1}{4+\sqrt{2}}y_2 -\frac{1}{6}x_1 y_1\\ \\
\frac {dy_2}{dt}(t)=-\alpha_2 y_2 + \frac{\beta_1}{2}x_1 y_1+\frac{\beta_3}{2}x_3 y_3= -\frac{2}{4+\sqrt{2}}y_2+\frac{1}{6}x_1 y_1+\frac{1}{6}x_3 y_3\\ \\
\frac {dy_3}{dt}(t)=\frac{\alpha_2}{2}y_2-\frac{\beta_3}{2}x_3 y_3= \frac{1}{4+\sqrt{2}}y_2 - \frac{1}{6}x_3 y_3 
\end{array}\right.
\end{equation}\\

The system in (\ref{numericsyst}) has at least two steady-states. 
Thanks to the symmetry of the coefficients, $(x_1,x_2,x_3,y_1,y_2,y_3)=(3, 2+\frac{1}{\sqrt{2}},1,1, 2+\frac{1}{\sqrt{2}},3)$ and $(x_1,x_2,x_3,y_1,y_2,y_3)=(1, 2+\frac{1}{\sqrt{2}},3,3, 2+\frac{1}{\sqrt{2}},1)$ are both equilibrium points.
Furthermore, the positive solutions of the steady states equations are given by
\[
x_2=\frac{1}{6}(4x_1y_1+\sqrt{2}x_1 y_1), \quad y_2=\frac{1}{6}(4x_1y_1+\sqrt{2}x_1 y_1), \quad
y_3=\frac{x_1 y_1}{x_3}
\]
or
\[
y_2=\frac{1}{6}(4x_1y_1+\sqrt{2}x_1 y_1), \quad x_2=\frac{1}{6}(4x_1y_1+\sqrt{2}x_1 y_1), \quad
x_3=\frac{x_1 y_1}{y_3}
\]
\\

By imposing the conservation law on the sums of the $x_i$ and $y_i$, i.e. $\sum_{i=1}^3 x_i= X=6+\frac{1}{\sqrt{2}}$ and $\sum_{i=1}^3 y_i=Y=6+\frac{1}{\sqrt{2}}$, we get either
\[
x_2=\frac{(6.70711 -  x_1) x_1}{1.10819 + x_1}, \quad x_3=\frac{7.43278 - 1.10819 x_1}{1.10819 +x_1},
\]
\[
y_1=\frac{7.43278 - 1.10819 x_1}{1.10819 + x_1}, \quad y_2=\frac{(6.70711 - x_1) x_1}{1.10819 +x_1}, 
\]
\[
y_3=\frac{1.86329\dot 10^{-15} + 1.10819 x_1 +x_1^2}{1.10819 +x_1}
\]

or the following formulation
\[
x_1=\frac{7.43278 - 1.10819 y_1}{1.10819 + y_1}, \quad x_2=\frac{(6.70711 - y_1) y_1}{1.10819 +y_1}, \]
\[x_3=\frac{1.86329\dot 10^{-15} + 1.10819 y_1 +y_1^2}{1.10819 +y_1},
\]
\[
y_2=\frac{(6.70711 -  y_1) y_1}{1.10819 + y_1}, \quad y_3=\frac{7.43278 - 1.10819 y_1}{1.10819 +y_1} 
\] \\

The Jacobian matrix of the system (not evaluated yet at the equilibrium and not considering the conservation law on the $x_i$ and $y_i$) is given by
\begin{equation} \vec{J}(x_1,x_3,y_1,y_3)=
\left[\begin{array}{cccccc}
-\frac{y_1}{6} & \frac{1}{\sqrt{2}+4} & 0 & -\frac{x_1}{6} & 0 & 0 \\
\frac{y_1}{6} & -\frac{2}{\sqrt{2}+4}&\frac{y_3}{6} & \frac{x_1}{6} & 0 & \frac{x_3}{6} \\
0 & \frac{1}{\sqrt{2}+4} & -\frac{y_3}{6} & 0 & 0 & -\frac{x_3}{6} \\
-\frac{y_1}{6} & 0 & 0 & -\frac{x_1}{6} & \frac{1}{\sqrt{2}+4} & 0 \\
\frac{y_1}{6} & 0 & \frac{y_3}{6} & \frac{x_1}{6} & -\frac{2}{\sqrt{2}+4} & \frac{x_3}{6} \\
0 & 0 & -\frac{y_3}{6} & 0 & \frac{1}{\sqrt{2}+4} & -\frac{x_3}{6}
\end{array}\right]
\end{equation}

We get always three eigenvalues with negative real part $\lambda_1$, $\lambda_5$, $\lambda_6$ and three null eigenvalues $\lambda_2$, $\lambda_3$, $\lambda_4$.\\
The eigenvalues zero determine a space with dimension at most three in the six dimensional one. As for the other three eigenvalues, they take negative values for every positive value of $x_1$ and $y_1$ smaller then $X$ and $Y$. The dynamics of the system converges to the manifold generated by the eigenvectors corresponding to the zero value eigenvalues.\\
It is not clear what is the behaviour on the stable manifold, but one could check, by numerically simulating the solutions of the system above for many initial conditions, that each point in the space generated by the eigenvectors corresponding to the null eigenvalues is an equilibrium point of the system. Therefore, we could conclude that on this space we have infinite not isolated fixed points, which are stable, but not asymptotically stable.

\thanks{
{\bf Acknowledgements}\\
This research was funded by the Academy of Finland, Centre of Excellence in Analysis and Dynamics Research and supported by a public grant as part of the Investissement d'avenir project, reference ANR-11-LABX-0056-LMH, LabEx LMH. 
Part of this research was performed during the participation of three of the authors at the thematic research program "Mathematical Biology" organised by the Mittag-Leffler Institute.
We also thank the two anonymous reviewers for their constructive remarks. }

%
%
%

\end{document}